\documentclass[a4paper, 10pt,reqno]{amsart}
\usepackage{amssymb}
\usepackage{amsmath}
\usepackage{amsthm}
\usepackage{latexsym}
\usepackage{mathrsfs}
\usepackage{bm}
\usepackage{eucal}
\usepackage{tipa}
\usepackage{slashed}
\usepackage{mathtools}
\usepackage{enumitem}

\usepackage{tikz-cd}

\usepackage{color}

\usepackage[boxsize=0.5em,aligntableaux=top]{ytableau}

\usepackage[numbers]{natbib}

\usepackage{hyperref}

\newcounter{mnotecount}[section]

\allowdisplaybreaks[2]

\theoremstyle{plain}
\newtheorem{theorem}{Theorem}
\newtheorem{proposition}[theorem]{Proposition}
\newtheorem{lemma}[theorem]{Lemma}

\newtheorem*{conjecture}{Conjecture}

\newtheorem{definition}[theorem]{Definition}
\newtheorem{remark}[theorem]{Remark}
\newtheorem{example}[theorem]{Example}

\DeclareMathOperator{\tho}{\text{\textthorn}}
\DeclareMathOperator{\edt}{\eth}

\renewcommand{\c}{\nabla}
\newcommand{\C}{\mathcal{C}}

\newcommand{\dd}{\mathrm{d}}

\renewcommand{\c}{\nabla}

\newcommand{\e}{\epsilon}

\renewcommand{\l}{\lambda}

\newcommand{\w}{\omega}

\newcommand{\T}{\Theta}

\newcommand{\mc}{\mathcal}

\newcommand{\tT}{\tilde{\Theta}}

\numberwithin{equation}{section}

\title{On the geometry of Petrov type II spacetimes}

\author[S. Aksteiner]{Steffen Aksteiner}
\email{steffen.aksteiner@aei.mpg.de}
\address{Albert Einstein Institute, Am M\"uhlenberg 1, D-14476 Potsdam, Germany }

\author[L. Andersson]{Lars Andersson}

\email{lars.andersson@aei.mpg.de}
\address{Albert Einstein Institute, Am M\"uhlenberg 1, D-14476 Potsdam, Germany }

\author[B. Araneda]{Bernardo Araneda}

\email{baraneda@famaf.unc.edu.ar}
\address{ Facultad de Matem\'atica, Astronom\'ia, F\'isica y Computaci\'on,
Universidad Nacional de C\'ordoba,
Instituto de F\'isica Enrique Gaviola, CONICET,
Ciudad Universitaria, (5000) C\'ordoba, Argentina}

\author[B. Whiting]{Bernard Whiting}

\email{bernard@phys.ufl.edu}
\address{Department of Physics, University of Florida, 2001 Museum Road, Gainesville, FL 32611-8440, USA}

\begin{document}

\begin{abstract}
In general, geometries of Petrov type II do not admit symmetries in terms of Killing vectors or  spinors. We introduce a weaker form of Killing equations which do admit solutions. In particular, there is an analog of the Penrose-Walker Killing spinor. Some of its properties, including associated conservation laws, are discussed.
Perturbations of Petrov type II Einstein geometries in terms of a complex scalar Debye potential yield complex solutions to the linearized Einstein equations. The complex linearized Weyl tensor is shown to be half Petrov type N. 
The remaining curvature component on the algebraically special side is reduced to a first order differential operator acting on the potential. 
\end{abstract}

\maketitle

\section{Introduction}

A remarkable property of vacuum spacetimes of Petrov type D is the existence of `hidden symmetries', 
namely, appropriate generalizations of Killing vectors such as Killing tensors and conformal Killing-Yano tensors. 
Penrose and Walker have shown \citep{WalkerPenrose} the existence of a valence 2 Killing spinor, from which 
one can obtain the above mentioned symmetries.
Because of the Goldberg-Sachs theorem, the vacuum type D condition is equivalent to the existence of 
two independent null geodesic congruences that are shear-free.
On the other hand, in general, vacuum spacetimes of Petrov type II do not possess any symmetries or hidden symmetries.
However they do admit a shear-free null geodesic congruence.
In this paper we show that this can be used to define a weaker version of the Killing equations, and we show that 
they are solved by a generalization of the Penrose-Walker Killing spinor.

The construction of solutions to the field equations for linear spinning fields in terms of scalar, tensorial or spinorial potentials 
has a long history and has been widely studied. In the case of the Maxwell field, the main names associated to this procedure 
are Debye and Hertz. Although the usage in the literature is not consistent, see Stewart \cite{1979Stewart}, we shall here refer to scalar potentials of the above 
mentioned type as Debye potentials. In this paper we focus on the spin-2 case, and consider the construction of solutions to the 
linearized Einstein equations on backgrounds of Petrov type II, in terms of Debye potentials. The analogous construction on 
backgrounds of Petrov type D, including the Schwarzschild and Kerr spacetimes has been widely studied, 
see e.g. \cite{CK1979, Wald, Chrzanowski:1975},
and plays an important role in the study of black hole perturbations \cite{Lousto:2002}, \cite{Hollands:2019} and the self-force problem \cite{2016PhRvD..94j4066M}, \cite{Barack:2017oir}. 
The construction of solutions to the linearized Einstein equation by the Debye potential method on backgrounds of Petrov type II is possible and is closely analogous to the type D case. 
Kegeles and Cohen \cite{CK1979} and Stewart \cite{1979Stewart} have carried out a systematic study of the Debye potential construction in this case, and in particular, Stewart calculated the tetrad components of the linearized Weyl tensor. In this work we show that the linearized Weyl tensor is half type N\footnote{The linearized metric generated from a Debye potential is naturally complex and therefore there are two Petrov classifications and half types possible, see below for details.}.

It is worth recalling briefly the situation for vacuum type D spacetimes.
We shall use the conventions and notation of \cite{PR1, PR2}.
A spacetime of Petrov type D admits two repeated principal null directions $l^a, n^a$, 
which have corresponding principal spin dyad $o^A, \iota^A$.
In terms of this principal frame, the only non-vanishing Weyl scalar is $\Psi_2$.
As was shown in \cite{WalkerPenrose}, vacuum spacetimes (for the more general situation see \cite{2014Andersson}) of Petrov type D admit a Killing spinor of the form
\begin{equation}\label{KStypeD}
 \mathring{K}_{AB}=\Psi^{-1/3}_2 o_{(A}\iota_{B)}
\end{equation}
(the ring `` $\mathring{}$ '' here is notation only intended to separate the type D case from the type II case that we discuss below),
which satisfies the equation
\begin{equation}\label{KSequation}
 \nabla_{A'(A}\mathring{K}_{BC)}=0.
\end{equation}
Different kinds of symmetries are associated to the object \eqref{KStypeD}, see \cite{1973hughston}, \cite[Section 6.7]{PR2}.
In particular in vacuum, the vector field defined by 
\begin{equation}
 \mathring{\xi}_{AA'}=\nabla^{B}{}_{A'}\mathring{K}_{AB}
\end{equation}
is a Killing vector, $\nabla_{(a}\mathring{\xi}_{b)}=0$. In the Kerr-NUT class, from \eqref{KStypeD} one can also construct a 
Killing tensor $H_{ab}$ and a second Killing vector $\eta_a=H_{ab}\mathring{\xi}^{b}$, see \cite{1973hughston} for details.

Consider an Einstein spacetime, i.e. the Einstein tensor being proportional to the metric or equivalently vacuum with cosmological constant, of Petrov type~II. By \cite{jeffryes1984}, a spacetime with a Killing spinor of valence 2, cf. eq. \eqref{KSequation}, 
has Weyl tensor of type D, N or O, so a type~II geometry does not admit a Killing spinor of valence~2.
Instead, we prove in section \ref{sec:parallelspinors} below: 

\begin{theorem} 
Let $(\mathcal{M}, g_{ab})$ be a real Einstein spacetime of Petrov type II, and let $o_A, \iota_A$ be a spin dyad, 
such that $o_A$ is a repeated principal spinor. Let $\Psi_i$ be the corresponding Weyl scalars. 
Define 
\begin{align}
 K_{AB} ={}& \Psi^{-1/3}_2 o_{(A}\iota_{B)} - \tfrac{1}{3}\Psi^{-4/3}_2\Psi_3 o_{A} o_{B}, \label{eq:ProjKs} \\
 \xi_{AA'} ={}& \nabla^B{}_{A'} K_{AB}. \label{eq:xiDef}
\end{align}
Then
\begin{enumerate} 
\item $K_{AB}$ solves the ``projected'' Killing spinor equation 
\begin{equation}\label{PKSeq}
 o^A\nabla_{A'(A}K_{BC)}=0.
\end{equation}
\item $\xi_{AA'}$ solves  the ``projected'' Killing equation
\begin{align} 
o^A (\nabla_{AA'}\xi_{BB'}+\nabla_{BB'}\xi_{AA'}) = 0.
\end{align}
\end{enumerate} 
\end{theorem} 

Now, the Debye potential construction produces {\em complex} solutions to the linearized Einstein equations, 
which leads to the possibility of having ``half types'' for the Weyl tensor. 
Recall first that in a {\em real}, four-dimensional orientable manifold with a metric of Lorentzian signature, 
the Hodge star operator $*$ acting on 2-forms satisfies $*^{2}=-1$, thus it has eigenvalues $\pm i$. 
As a consequence, the eigenspaces of $*$ are complex, i.e.
a self-dual (SD) or anti-self-dual (ASD) 2-form is necessarily complex.
Any real 2-form can be written as the sum of a SD part and an ASD part, and these pieces are complex conjugates of each other.
Since the Weyl tensor $C_{abcd}$ can be regarded as a tensor-valued 2-form, a similaR{\'o}zgar discussion applies to it.
Namely, $C_{abcd}$ can be written as the sum of a SD piece $C^{+}_{abcd}$ and an ASD piece $C^{-}_{abcd}$, 
where $C^{+}_{abcd}$ and $C^{-}_{abcd}$ are complex conjugates of each other.
In spinor terms (see section \ref{sec:Preliminaries} for notation and details), these are 
$C^{-}_{abcd}=\Psi_{ABCD}\bar\epsilon_{A'B'}\bar\epsilon_{C'D'}$ 
and  $C^{+}_{abcd}=\bar\Psi_{A'B'C'D'}\epsilon_{AB}\epsilon_{CD}$, where $\Psi_{ABCD}$ is the Weyl curvature spinor, 
see \cite[Section 4.6]{PR1}.

On the other hand, if the spacetime metric is {\em complex}, the above decomposition still holds, 
but the pieces $C^{+}_{abcd}$ and $C^{-}_{abcd}$ are now independent entities, 
so one has two independent Weyl spinors $\Psi_{ABCD}$ and $\tilde{\Psi}_{A'B'C'D'}$, 
see \cite[Section 6.9]{PR2}.
In particular, the Petrov types of $\Psi_{ABCD}$ and $\tilde{\Psi}_{A'B'C'D'}$ are independent; for example, one part may be algebraically 
special while the other one is algebraically general\footnote{Note however that due to results by R{\'o}zga \cite{1977RpMP...11..197R}, in order for a complex 4-dimensional spacetime to admit a real, Lorentzian slice, the algebraic type of $\Psi_{ABCD}$ and $\tilde\Psi_{ABCD}$ must be the same. See also \cite{1977IJTP...16..663W}.}.  
This also applies to linearized gravity, where even if the background metric is real, a complex perturbation will in general have independent 
SD and ASD linearized curvatures.
In what follows we denote the SD and ASD linearized curvature spinors for a complex perturbation by $\dot{\tilde{\Psi}}_{A'B'C'D'}$ 
and $\dot\Psi_{ABCD}$ respectively, see section \ref{sec:Preliminaries} for details.

We shall make use of the Geroch-Held-Penrose (GHP) formalism \cite{GHP}.
For algebraically special spacetimes,
the method of adjoint operators introduced by Wald in \cite{Wald} (see section \ref{sec:Debye} for a brief review) can be used to show that 
if a scalar field $\chi$ of GHP weight $\{-4,0\}$ solves the Debye equation
\begin{align} 
\left(\left( \tho' - \bar{\rho}' \right) \left(\tho + 3 \rho \right) - \left(\edt' - \bar{\tau}\right) \left(\edt + 3 \tau \right)  - 3 \Psi_2 \right) \chi = 0,
\end{align}
then the complex tensor field $h_{ab}=\mathcal{S}^{\dag}(\chi)_{ab}$, where $\mathcal{S}^{\dag}$ is the adjoint of the 
operator $\mathcal{S}$ defined in eq. \eqref{eq:operatorS} below, is a solution to the linearized Einstein vacuum equations (possibly with cosmological constant).
For perturbations of vacuum type D spacetimes,
it was shown in \cite{CK1979} that the ASD linearized Weyl spinor $\dot{\Psi}_{ABCD}$ has special algebraic structure, \textit{viz.} it is of Petrov type N,
\begin{equation}\label{ASDlweyl-typeD}
 \dot{\Psi}_{ABCD}=o_{A}o_{B}o_{C}o_{D} \dot{\Psi}_{4},
\end{equation}
whereas the SD linearized Weyl spinor may be algebraically general. 

Furthermore, while the linearized Weyl scalars associated to a Debye potential $\mathring{\chi}$ are given in general by fourth order differential 
operators applied to $\mathring{\chi}$, it was shown in \cite{CK1979}, see also \cite{AB2019}, that the scalar field $\dot{\Psi}_{4}$ 
in \eqref{ASDlweyl-typeD} is given by the simple expression,
\begin{equation}\label{dpsi4-typeD}
 \dot{\Psi}_{4} = c \mathcal{L}_{\mathring{\xi}}\mathring{\chi},
\end{equation}
with (possibly complex) constant $c$. Here, $ \mathcal{L}_{\mathring{\xi}}$ is the Lie derivative along $\mathring{\xi}^a$, which itself is given by \eqref{eq:xiDef} for the type D Killing spinor. In the Kerr spacetime, \eqref{dpsi4-typeD} is essentially the time derivative of $\mathring{\chi}$, see also \cite{Lousto:2002}. For Petrov type D, this reduction was done by Kegeles and Cohen \cite{CK1979} for the vacuum case and by Torres del Castillo \cite{TorresdelCasillo:1994} including a cosmological constant\footnote{One of the Authors, B.W., also did this tedious computation including a cosmological constant in 1983, but did not publish it.}.

One may think that the remarkably simple structure \eqref{ASDlweyl-typeD}-\eqref{dpsi4-typeD} is 
associated to the very special symmetry properties of vacuum type D spacetimes, i.e. to the existence of the 
``hidden'' symmetry \eqref{KStypeD} and the two associated isometries mentioned before.
In this note we generalize these results to Petrov type II spacetimes, which in general do not possess any isometries:

\begin{theorem} \label{thm}
Consider an Einstein spacetime of Petrov type II with repeated principal spinor $o^A$.
Let $h_{ab}$ be a complex solution to the linearized Einstein vacuum equations
generated by a Debye potential $\chi$. Then
\begin{enumerate}
\item the ASD Weyl spinor of $h_{ab}$ is of Petrov type N,
\begin{equation}\label{eq:linearizedWeyl}
\dot\Psi_{ABCD} = o_{A}o_{B}o_{C}o_{D} \dot\Psi_{4}.
\end{equation}
\item the non-vanishing component of \eqref{eq:linearizedWeyl} is given by
\begin{align}
\dot\Psi_{4}={}& -(\Psi_2^{4/3} \xi^a \Theta_a + 3 \Psi_{2}^2 + 6 \Psi_2 \Lambda) \chi.
\end{align}
Here $\xi^a$ is given by \eqref{eq:xiDef}, where $K_{AB}$ is the projected Killing spinor \eqref{eq:ProjKs}, $\Theta_a$ is the GHP connection and $6\Lambda$ corresponds to the cosmological constant.
\end{enumerate}
\end{theorem}

\begin{remark}
Upon finishing this work, we found a virtually unknown preprint by Jeffryes, \cite{Jeffryes:1986}, about half-algebraically special geometries and potentials for field equations. Although Theorem~\ref{thm} can alternatively be obtained from \S 8 of that work, the interpretation of the derivative in terms of a projected Killing vector has not been given there. In fact the focus of that work was on the non-linear case coupled to Yang-Mills and the relation to the situation here is quite intricate, see remark~\ref{rem:RelatedWork} for further details.
\end{remark}

In the case of linearized gravity in Minkowski spacetime, it was shown in \cite{torres1999} that real solutions of the linearized Einstein vacuum equations 
are in one-to-one correspondence with complex solutions with half-flat curvature, 
which in turn are in one-to-one correspondence with solutions of the scalar wave equation.
The result of Theorem \ref{thm} tells us that, while the complex metric perturbation generated by a Debye potential in a type II space is not half-flat, the linearized curvature has a simple structure since it is half type N. 

For vacuum type D spacetimes, it is sometimes assumed that, up to gauge, all real solutions of the linearized Einstein vacuum equations can be obtained, locally, as the real part of a metric generated by a Debye potential; for recent advances in the Schwarzschild and Kerr cases see respectively \cite{2018:Prabhu:Wald} and \cite{Hollands:2019}.
For the more general vacuum type II case, from these considerations we expect the result of Theorem \ref{thm} to be 
of relevance for addressing the following conjecture:
\begin{conjecture}
 All real solutions of the linearized Einstein vacuum equations on a vacuum type II background can be locally obtained, up to gauge, as the real part of the metric generated by a Debye potential.
\end{conjecture}
We also point out, that Jeffryes in \cite{Jeffryes:1986} made remarks supporting the validity of this conjecture. 
\begin{remark}
Metrics generated from a Debye potential are always in radiation gauge. 
It is known that any perturbation of Petrov type II geometries can be transformed into radiation gauge, see \cite{Whiting:2007}.
\end{remark}
Theorem \ref{thm} generalizes the result \eqref{ASDlweyl-typeD}-\eqref{dpsi4-typeD} in the type D case to type II.
We stress that a generic type II spacetime does not possess any ordinary or hidden symmetries, but only the more general objects \eqref{eq:ProjKs} and \eqref{eq:xiDef} introduced in this work.
While it can be shown that the ordinary, valence-2 Killing spinor equation is equivalent to the real, conformal Killing-Yano equation, 
which is itself a generalization to differential forms of the conformal Killing equation, 
for the `projected' Killing spinor equation \eqref{PKSeq} no such equivalence exists. 
We will give a geometric interpretation to the origin of \eqref{PKSeq} in terms of spinors that are parallel under a suitable connection especially adapted to the geometry, which is  the conformal-GHP connection. This also allows us to generalize the result \eqref{eq:ProjKs}-\eqref{PKSeq} to non-vacuum spacetimes in the real-analytic case,
and to derive conservation laws associated to projected Killing spinors.

In appendix~\ref{app:RobinsonTrautman}, we also review the Robinson-Trautman reduction of the Einstein equations which admits solutions of various Petrov types. In particular there is a Petrov type II solution which we use in example~\ref{ex:xiForRT} to compute the projected Killing vector $\xi^a$. 

Most computations were performed with Spinframes \cite{SpinFrames}, based on the symbolic computer algebra package xAct for Mathematica.

\section{Preliminaries}\label{sec:Preliminaries}

\subsection{The 2-spinor formalism}

In this paper we shall make extensive use of the 2-spinor formalism, following the notation and conventions in \cite{PR1, PR2}.
The spinor bundles $\mathbb{S}\to\mathcal{M}$ and $\mathbb{S}'\to\mathcal{M}$ are rank-2 vector bundles 
with symplectic forms $\epsilon_{AB}$ and $\bar\epsilon_{A'B'}$, such that 
$T\mathcal{M}\otimes\mathbb{C}\cong\mathbb{S}\otimes\mathbb{S}'$ and $g_{ab}=\epsilon_{AB}\bar\epsilon_{A'B'}$.
The spaces of SD and ASD 2-forms are written in spinor terms as
$\Lambda^{2}_{+}\cong\mathbb{S}'^{*}\odot\mathbb{S}'^{*}$ and $\Lambda^{2}_{-}\cong\mathbb{S}^{*}\odot\mathbb{S}^{*}$ 
(with $\mathbb{S}^{*}$ the dual of $\mathbb{S}$, etc.);
in other words, a real 2-form $F_{ab}=F_{[ab]}$ has the spinor decomposition 
\begin{equation}\label{2form}
 F_{ab}=\phi_{AB}\bar\epsilon_{A'B'}+\bar{\phi}_{A'B'}\epsilon_{AB},
\end{equation}
where $F^{-}_{ab}=\phi_{AB}\bar\epsilon_{A'B'}$ and $F^{+}_{ab}=\bar{\phi}_{A'B'}\epsilon_{AB}$ are the SD and ASD 
parts of $F_{ab}$ respectively, and $\phi_{AB}=\phi_{(AB)}$.
The Riemann tensor admits a similar decomposition \cite[Eq. (4.6.38)]{PR1}:
\begin{align}
\nonumber R_{abcd} ={}& \Psi_{ABCD}\bar\epsilon_{A'B'}\bar\epsilon_{C'D'}+\bar\Psi_{A'B'C'D'}\epsilon_{AB}\epsilon_{CD} \\
\nonumber & +\Phi_{ABC'D'}\bar\epsilon_{A'B'}\epsilon_{CD}+\Phi_{A'B'CD}\epsilon_{AB}\bar\epsilon_{C'D'} \\
  & +2\Lambda (\epsilon_{AC}\epsilon_{BD}\bar\epsilon_{A'C'}\bar\epsilon_{B'D'}-\epsilon_{AD}\epsilon_{BC}\bar\epsilon_{A'D'}\bar\epsilon_{B'C'}),
  \label{Riemann}
\end{align}
where $\Psi_{ABCD}=\Psi_{(ABCD)}$ is the Weyl conformal spinor, $\Phi_{ABC'D'}=\Phi_{(AB)(C'D')}$ is the trace-free 
Ricci spinor (which is real, i.e. $\bar\Phi_{ab}=\Phi_{ab}$), and $\Lambda=R/24$ represents the scalar curvature 
(which is also real, $\bar\Lambda=\Lambda$)\footnote{In the Einstein case, we have $\Phi_{ABC'D'}=0$ and cosmological constant $6\Lambda$.}.

\subsection{Self-duality}

Let $(\mathcal{M}, g_{ab})$ be a real, four-dimensional Lorentzian manifold. Throughout we make use of the abstract index notation.
We assume the spacetime to be orientable so that there is a volume form $\varepsilon_{abcd}$ and the associated 
Hodge star $*:\Lambda^{k}\to\Lambda^{4-k}$, where $\Lambda^{k}$ is the space of $k$-forms. 
For 2-forms, $*$ satisfies $*^{2}=-1$; this induces a decomposition $\Lambda^2=\Lambda^{2}_{+}\oplus\Lambda^{2}_{-}$, 
where $\Lambda^{2}_{+}$ (resp. $\Lambda^{2}_{-}$) is the rank 3 eigenbundle of $*$ associated to the eigenvalue $+i$ (resp. $-i$). 
Elements of $\Lambda^{2}_{+}$ are called self-dual (SD) 2-forms, and those of $\Lambda^{2}_{-}$ are 
anti-self-dual (ASD) 2-forms.
Since the Riemann curvature tensor has the symmetries $R_{abcd}=R_{[ab][cd]}$, one can also apply the Hodge duality operation to it,
in particular this defines the left- and right-dual Riemann tensors by ${}^{*}R_{abcd}=\frac{1}{2}\varepsilon_{ab}{}^{ef}R_{efcd}$ 
and $R^{*}_{abcd}=\frac{1}{2}\varepsilon_{cd}{}^{ef}R_{abef}$, respectively.
For the Weyl tensor, the left- and right-duals coincide: ${}^{*}C_{abcd}=C^{*}_{abcd}$. 
One then defines the SD and ASD Weyl tensors by
\begin{equation}\label{SD-ASDWeylTensor}
 C^{\pm}_{abcd}:=\tfrac{1}{2}(C_{abcd}\mp i{}^{*}C_{abcd}).
\end{equation}
These tensors satisfy $*C^{\pm}_{abcd}=\pm i  C^{\pm}_{abcd}$.

In spinor terms, the SD and ASD Weyl tensors \eqref{SD-ASDWeylTensor} are
\begin{equation}
 C^{-}_{abcd}=\Psi_{ABCD}\bar\epsilon_{A'B'}\bar\epsilon_{C'D'}, \qquad C^{+}_{abcd}=\bar\Psi_{A'B'C'D'}\epsilon_{AB}\epsilon_{CD}.
\end{equation}
In particular, the SD and ASD Weyl tensors are complex conjugate of each other, 
which follows from \eqref{SD-ASDWeylTensor}.

\subsection{Complex geometries}

For complex geometries, the above decompositions still apply, but now a pair of complex conjugate quantities that appear 
together in a real expression such as \eqref{2form}, is replaced by two independent entities; see \cite[Section 6.9]{PR2}.
For example, for a complex 2-form, the spinor $\bar\phi_{A'B'}$ in \eqref{2form} is replaced by a spinor $\psi_{A'B'}$ which 
is no longer the complex conjugate of $\phi_{AB}$.
Similarly, the Riemann tensor of a complex metric has a spinor decomposition analogous to \eqref{Riemann}, 
but where $\bar\Psi_{A'B'C'D'}$ is now replaced by another spinor $\widetilde{\Psi}_{A'B'C'D'}$ which is no longer the 
complex conjugate of $\Psi_{ABCD}$.
The Ricci spinor $\Phi_{ABC'D'}$ and the scalar curvature $\Lambda$ do not acquire ``tilded'' versions because 
of the original reality conditions $\bar\Phi_{ab}=\Phi_{ab}$ and $\bar\Lambda=\Lambda$, which are a consequence of 
the symmetries of the Riemann tensor; they simply become complex objects.

The fact that $\Psi_{ABCD}$ and $\widetilde{\Psi}_{A'B'C'D'}$ are now two independent entities implies that 
one correspondingly has two independent algebraic classification schemes for the Weyl curvature spinors, 
so we can have for example conformally `half-flat' manifolds if, say, $\Psi_{ABCD}=0$ and $\widetilde{\Psi}_{A'B'C'D'}\neq0$, 
or `half-algebraically special' solutions if $\Psi_{ABCD}$ is algebraically special while $\widetilde{\Psi}_{A'B'C'D'}$ is general. These remarks about complex geometries apply of course also to complex, linear perturbations of real geometries.

\subsection{Conformal-GHP connections}

In section \ref{sec:parallelspinors} we shall discuss some geometric aspects of projected Killing spinors. In order to do this we need an extension of the GHP derivative that includes conformal transformations.
We call such an extension the `conformal-GHP' connection, and give a brief review of some aspects of the construction that are relevant for this work.

Let $(\mathcal{M}, g_{ab})$ be a Lorentzian spacetime, with Levi-Civita connection $\nabla_{a}$.
Let $(o^A,\iota^A)$ be a general spin dyad. We allow two kinds of transformations: 
$(i)$ `GHP transformations' $o^A\to\gamma o^A$, $\iota^A\to\gamma^{-1}\iota^A$, with $\gamma$ a non-vanishing complex scalar field, and 
$(ii)$ conformal transformations $o^A\to o^A$, $\iota^A\to\Omega^{-1}\iota^A$, where $\Omega$ is a positive scalar field.
The transformation $(ii)$ is induced by a transformation of the spin metric $\epsilon^{AB}\to\Omega^{-1}\epsilon^{AB}$, 
which is in turn induced by a conformal transformation of the metric, $g^{ab}\to\Omega^{-2}g^{ab}$.

Let $\mathbb{S}_{\{p,q\}}[w]$ be the vector bundle of conformally weighted spinors with GHP weight 
$\{p,q\}$ and conformal weight $w$, and with an arbitrary index structure.
This means that a section $\varphi^{A...A'...}_{B...B'...}\in\Gamma(\mathbb{S}_{\{p,q\}}[w])$ transforms as 
$\varphi^{A...A'...}_{B...B'...} \to \gamma^p \bar\gamma^q\varphi^{A...A'...}_{B...B'...}$ under GHP transformations, and as 
$\varphi^{A...A'...}_{B...B'...} \to \Omega^{w} \varphi^{A...A'...}_{B...B'...}$ under conformal transformations.
For example, we have $o^A\in\Gamma(\mathbb{S}_{\{1,0\}}[0])$ and $\iota^A\in\Gamma(\mathbb{S}_{\{-1,0\}}[-1])$.
The conformal-GHP covariant derivative is a linear connection
$\mathcal{C}_{AA'}:\Gamma(\mathbb{S}_{\{p,q\}}[w]) \to \Gamma(T^{*}\mathcal{M}\otimes\mathbb{S}_{\{p,q\}}[w])$.
In this work all objects considered have vanishing `$q$-weight', so we shall restrict to quantities of GHP weight $\{p,0\}$.
For the general case and more details see \cite{2018:Araneda} and references therein\footnote{When acting on {\em scalar} quantities, 
the projection of $\mathcal{C}_{AA'}$ on a null tetrad reduces to the conformally invariant GHP operators of 
Penrose and Rindler \cite[Eq. (5.6.36)]{PR1}. 
We note however that, for general spinors and tensors, the operators in \cite[Eq. (5.6.36)]{PR1} do not map conformal densities 
to conformal densities.}.
The action of $\mathcal{C}_{AA'}$ on, say, a spinor field $\varphi^{BB'}_{CC'}$ with GHP weight $\{p,0\}$ and conformal weight $w$, is given by
\begin{align}
\mathcal{C}_{AA'}\varphi^{BB'}_{CC'}={}& \nabla_{AA'}\varphi^{BB'}_{CC'}+(wf_{AA'}+p(\omega_{AA'}+B_{AA'}))\varphi^{BB'}_{CC'} \nonumber \\
 &+\epsilon_{A}{}^{B}f_{QA'}\varphi^{QB'}_{CC'}+\epsilon_{A'}{}^{B'}f_{AQ'}\varphi^{BQ'}_{CC'}-f_{CA'}\varphi^{BB'}_{AC'}-f_{AC'}\varphi^{BB'}_{CA'},
 \label{C}
\end{align}
with
\begin{subequations}
\begin{align}
 & \omega_{a}:=-\epsilon n_a+\epsilon'\ell_a-\beta' m_a+\beta\bar{m}_a, \label{omega} \\
 & B_{a}:=-\rho n_a + \tau \bar{m}_a, \label{formB} \\
 & f_{a}:= \rho n_a + \rho'\ell_a-\tau' m_a-\tau\bar{m}_a, \label{f}
\end{align}
\end{subequations}
where we are using standard GHP notation for spin coefficients.
For spinors with a different index structure, the corresponding action of $\mathcal{C}_{AA'}$ can be deduced from 
\eqref{C} by linearity and the Leibniz rule.
By construction, the connection \eqref{C} is covariant under combined conformal 
and GHP transformations,
\begin{equation*}
 \mathcal{C}_a \varphi^{B...B'...}_{C...C'...} \to \gamma^p \Omega^w \mathcal{C}_a \varphi^{B...B'...}_{C...C'...}.
\end{equation*}
In particular, applying \eqref{C} to the spin frame $(o_A,\iota_A)$ and taking into account that 
$o_A\in\Gamma(\mathbb{S}_{\{1,0\}}[1])$ and $\iota_A\in\Gamma(\mathbb{S}_{\{-1,0\}}[0])$, one finds
\begin{subequations} 
\begin{align}
 & \mathcal{C}_{AA'}o_B=(o^Co^D\nabla_{CA'}o_D)  \iota_A\iota_B,  \label{Co}\\
 & \mathcal{C}_{AA'}\iota_B=(\iota^C\iota^D\nabla_{CA'}\iota_D) o_Ao_B.  \label{Ci}
\end{align}
\end{subequations}

\begin{remark}\label{remark-dyad}
Analogously to the usual GHP connection, the conformal-GHP connection $\mathcal{C}_{AA'}$ depends on the choice 
of a spin dyad $(o^A,\iota^A)$. In what follows, the dyad $(o^A,\iota^A)$  will always be understood to 
be the one associated to $\mathcal{C}_{AA'}$.
\end{remark}

\section{Special geometry and parallel spinors}\label{sec:parallelspinors}

In this section we provide a geometric interpretation for projected Killing spinors and discuss general properties such as integrability conditions, the solution space  and a relation to conservation laws. Except for some specific examples, we do not assume the geometry to be Ricci-flat.
Several results below need a complexification of spacetime and in those situations we impose real-analyticity; this will be explicitly indicated in each case.

We will first show that the projected Killing spinor equation can be understood as a consequence of 
the existence of spinors that are parallel under the conformal-GHP connection.

\begin{remark}
The formulation in this section is conformally invariant (except for the specific examples in which we assume the Einstein condition, see (2) in Lemma~\ref{lemma-f} and examples \ref{ex:KvFromKs} and \ref{example-spinlowering}). This means that all results below are valid not only for a specific metric but for the equivalence class of metrics conformally related to each other, i.e. for conformal structures. For simplicity, however, we shall state the results in terms of spacetimes and not conformal structures.
\end{remark}

Let $(\mathcal{M}, g_{ab})$ be a Lorentzian spacetime, and let $o^{A}$ be a spinor field satisfying
\begin{align}\label{SFR}
 o^A o^B \nabla_{AA'} o_B = 0.
\end{align}
Such an $o^A$ will be called a {\em shear-free ray} (SFR). 
If $o_A$ is SFR, it follows from \eqref{Co} that it is parallel under $\mathcal{C}_{AA'}$. 
In GHP notation, \eqref{SFR} is equivalent to $\kappa=\sigma=0$. 
The condition $\kappa=0$ encodes that the null congruence associated to the vector field $o^{A}\bar{o}^{A'}$ 
is geodesic, whereas $\sigma=0$ states that this congruence is shear-free.
Thus, a non-trivial solution to \eqref{SFR} imposes the existence of a shear-free null geodesic congruence.

\begin{remark}
Several results below require a complexified spacetime, for which we need to impose real-analyticity.
The complex extension of $\mathcal{M}$ will be denoted $\mathbb{C}\mathcal{M}$. 
Our interest here is in the case where $\mathbb{C}\mathcal{M}$ arises as the complexification of a real spacetime with Lorentzian signature, see \cite[Section 6.9]{PR2}. 
\end{remark}

The reason why we need to complexify $\mathcal{M}$ is the following important result about SFRs:
\begin{lemma}[Proposition (7.3.18) in \cite{PR2}]
Let $o^A$ be a spinor, and let $\mu^{A'}$, $\nu^{A'}$ be a primed spin dyad.
Then the complex vector fields $X^a = o^A \mu^{A'}$, $Y^a = o^A \nu^{A'}$ on $\mathbb{C}\mathcal{M}$ are in involution 
if and only if $o^A$ is SFR.
\end{lemma}

Hence by Frobenius' theorem, the distribution defined by $\{X^a,Y^a\}$ is surface-forming in $\mathbb{C}\mathcal{M}$.
The surfaces associated to this distribution are {\em complex}, and they are called {\em $\beta$-surfaces}, 
see e.g. \citep[pp. 309-310]{PR2}:

\begin{definition}[$\beta$-surface]
A $\beta$-surface is a complex 2-dimensional surface in $\mathbb{C}\mathcal{M}$ whose tangent vectors at any one point 
are all of the form $o^A\mu^{A'}$ for some $\mu^{A'}$, where $o^A$ is fixed and satisfies equation \eqref{SFR}. 
\end{definition}

For the following it is convenient to introduce some additional notation.
\begin{definition}\label{Definition-projectedC}
Let $\mathcal{C}_{AA'}$ be the conformal-GHP connection associated to a spin dyad $(o^A,\iota^A)$.
We define the operators
\begin{equation}
 \tilde{\mathcal{C}}_{A'}:=o^A\mathcal{C}_{AA'}, \qquad \mathcal{C}_{A'}:=\iota^A\mathcal{C}_{AA'}. \label{tildeC}
\end{equation}
\end{definition}

It is worth discussing some properties of $\tilde{\mathcal{C}}_{A'}$. It is a linear operator 
that satisfies the Leibniz rule, and which maps a section of $\mathbb{S}_{\{p,0\}}[w]$ into a section of 
$\mathbb{S}'\otimes\mathbb{S}_{\{p+1,0\}}[w]$, where $\mathbb{S}'$ is the primed spin bundle introduced in 
section \ref{sec:Preliminaries}.
Consider now a $\beta$-surface $\Sigma$ in $\mathbb{C}\mathcal{M}$. Any element of the tangent bundle $T\Sigma$ is of the 
form $(x,o^A\mu^{A'})$, where $x\in\Sigma$, $o^A$ is fixed, and $\mu^{A'}$ is some primed spinor at $x$.
Therefore we can identify $T\Sigma$ with the restriction of the primed spin bundle to $\Sigma$, $\mathbb{S}'|_{\Sigma}$.
Similarly, the cotangent bundle $T^{*}\Sigma$ can be identified with the dual $\mathbb{S}'^{*}|_{\Sigma}$.
This means that a primed spinor field $\psi_{A'}$ can be thought of as a ``1-form'' in $T^{*}\Sigma$.
Therefore, restricting to $\beta$-surfaces, the operator $\tilde{\mathcal{C}}_{A'}$ is
$\tilde{\mathcal{C}}_{A'}:\Gamma(\mathbb{S}_{\{p,0\}}[w])\to\Gamma(T^{*}\Sigma\otimes\mathbb{S}_{\{p+1,0\}}[w])$, 
from which we see that it can be regarded as a connection, in the usual sense, on conformally and GHP-weighted 
vector bundles over $\beta$-surfaces on $\mathbb{C}\mathcal{M}$, see \cite{2020:Araneda}. 

We have the following result, which does not require analyticity:
\begin{lemma}\label{lemma-tildeC}
Let $(\mathcal{M}, g_{ab})$ be a Lorentzian spacetime with $\tilde{\mathcal{C}}_{A'}$ given by \eqref{tildeC}, where $o^A$ satisfies \eqref{SFR}. Then 
\begin{equation}\label{flat-tildeC}
 [\tilde{\mathcal{C}}_{A'}, \tilde{\mathcal{C}}_{B'}]=0
\end{equation}
if and only if $o^A$ is a repeated principal spinor of $\Psi_{ABCD}$, that is
\begin{equation}\label{typeII}
 \Psi_{ABCD}o^Bo^Co^D=0.
\end{equation}
\end{lemma}

This result follows from Lemmas 3.4 and 3.5 in \citep{2020:Araneda}. 
As mentioned, the proof of Lemma \ref{lemma-tildeC} does not require analyticity.
However, from the discussion above we know that {\em on $\beta$-surfaces}, which live in $\mathbb{C}\mathcal{M}$, we can interpret 
$\tilde{\mathcal{C}}_{A'}$ as a connection. In that case the commutator 
$[\tilde{\mathcal{C}}_{A'}, \tilde{\mathcal{C}}_{B'}]$ is the curvature of $\tilde{\mathcal{C}}_{A'}$. Thus, 
the result of Lemma \ref{lemma-tildeC} tells us that, as long as conditions \eqref{SFR} and \eqref{typeII} hold, 
{\em $\tilde{\mathcal{C}}_{A'}$ is a flat connection on $\beta$-surfaces}. 
From \cite[eq. (3.27)]{2020:Araneda} we see that one can associate a (twisted) de Rham complex to $\tilde{\mathcal{C}}_{A'}$.
Local exactness of a de Rham complex allows to find (local) potentials in specific situations; 
for example, we find the following:
\begin{lemma}\label{lemma-f}
Suppose that $o^A$ is SFR \eqref{SFR} and a repeated principal spinor \eqref{typeII}. 
\begin{enumerate}
\item Assume $(\mathcal{M},g_{ab})$ is real-analytic with complexification $\mathbb{C}\mathcal{M}$.
Then there exist scalar fields $\phi$ and $\eta$ on $\mathbb{C}\mathcal{M}$, whose GHP and conformal weights are $\{0,0\}$, $w=-1$ 
 and $\{-2,0\}$, $w=-1$ respectively, such that $f_a$, given in \eqref{f}, takes the form
 \begin{equation}
  f_{AA'}=\nabla_{AA'}\log\phi-o_A\tilde{\mathcal{C}}_{A'}\eta. \label{f2}
\end{equation}
\item Assume that $(\mathcal{M}, g_{ab})$ is real and that $g_{ab}$ is Einstein. Then, with Weyl scalars $\Psi_i$, \eqref{f2} holds for $\phi$ and $\eta$ given by
\begin{align}\label{phieta}
 \phi=\Psi^{1/3}_2, \qquad \eta=\tfrac{1}{3}\Psi^{-1}_2\Psi_3.
\end{align}
\end{enumerate}
\end{lemma}

The proof of eq. \eqref{f2}, along with other results and applications, will be given in a forthcoming publication.
This proof involves the existence of $\beta$-surfaces and that is the reason why it is formulated in the real-analytic setting.
On the other hand, \eqref{phieta} can be easily demonstrated in any Einstein spacetime, not necessarily analytic, by using the GHP form of the Bianchi identities\footnote{Since Bianchi identities are not conformally 
invariant, the expression \eqref{phieta} breaks conformal invariance.}.

\begin{remark}\label{remark-tildeC}
Consider the complexified spacetime $\mathbb{C}\mathcal{M}$.
\begin{itemize}
\item Since $\tilde{\mathcal{C}}_{A'}$ is a flat connection, the equation $\tilde{\mathcal{C}}_{A'}\lambda=0$ has 
non-trivial solutions for any weights $\{p,0\}$ and $w$. In particular, from \eqref{C} we see that,
choosing $w=p=-1$, we have
\begin{equation}
 0=\tilde{\mathcal{C}}_{A'}\lambda=o^{A}\partial_{AA'}\lambda-o^{A}\omega_{AA'}\lambda,
\end{equation}
where we used that $o^{A}f_{AA'}=-o^{A}B_{AA'}$. It then follows that such $\lambda$ satisfies
\begin{equation}\label{lambda}
 o^A\partial_{AA'}\log\lambda = o^A\omega_{AA'}.
\end{equation}
\item Let $\psi$ be a scalar field on $\mathbb{C}\mathcal{M}$ with GHP weight $\{p,0\}$ and conformal weight $w$.
Using \eqref{C}, \eqref{f2} and \eqref{lambda}, we see that the action of $\tilde{\mathcal{C}}_{A'}$ reduces to
\begin{equation}\label{identity-tildeC}
 \tilde{\mathcal{C}}_{A'}\psi=\phi^{-(w-p)}\lambda^{-p} o^A\partial_{AA'}[\phi^{w-p}\lambda^{p}\psi].
\end{equation}
\end{itemize}
\end{remark}

Using formula \eqref{identity-tildeC}, it follows that if the weighted scalar field $\psi$ satisfies $\tilde{\mathcal{C}}_{A'}\psi=0$, 
then defining $F=\phi^{w-p}\lambda^{p}\psi$, where $\phi$ and $\lambda$ are defined in \eqref{f2} and \eqref{lambda} respectively,
this is equivalent to 
\begin{equation}\label{constantbetas}
 o^A\partial_{AA'}F=0.
\end{equation}
On $\mathbb{C}\mathcal{M}$, these functions are said to be `constant on $\beta$-surfaces', since the operator $o^A\partial_{AA'}$ 
represents translations along the $\beta$-surfaces associated to $o^A$.
That is to say, there exist linearly independent spinor fields $\mu^{A'}$, $\nu^{A'}$ such that 
$o^A\mu^{A'}\partial_{AA'}\equiv \partial_{\tilde{z}}$ and $o^A\nu^{A'}\partial_{AA'}\equiv \partial_{\tilde{w}}$,
where $\tilde{z},\tilde{w}$ are coordinates {\em along} the $\beta$-surfaces.
Then equation \eqref{constantbetas} can be interpreted as 
\begin{equation}\label{constantbetas2}
 F=F(z,w), 
\end{equation}
where $z,w$ are coordinates {\em constant} on each $\beta$-surface. Notice that \eqref{constantbetas2} is 
a function on $\mathbb{C}\mathcal{M}$.

\begin{remark}\label{remark-iota}
Let $\iota^A$ be any spinor field independent of $o^A$. Then the equation 
\begin{equation}\label{iotanabla}
 \iota^A\partial_{AA'}f=0
\end{equation}
has only constant functions as solutions. To see this, note that if \eqref{iotanabla} holds we must also have $[\iota^A\nabla_{AA'},\iota^B\nabla_{BB'}]f=0$, but
\begin{equation}
 0=[\iota^A\nabla_{AA'},\iota^B\nabla_{BB'}]f=-\epsilon_{A'B'}(\iota^{A}\iota^{B}\nabla_{A}{}^{C'}\iota_{B})o^{C}\nabla_{CC'}f.
\end{equation}
By definition we know that $\iota^{A}\iota^{B}\nabla_{A}{}^{C'}\iota_{B}=0$ if and only if 
$\iota^A$ is SFR, but this is not the case since $\iota^A$ is arbitrary, so we get $o^C\nabla_{CC'}f=0$, which, 
when combined with \eqref{iotanabla}, implies $\nabla_{AA'}f=0$, i.e. $f$ is constant.
\end{remark}

In terms of this geometric setup, the projected Killing spinor equation can be deduced as follows:
\begin{proposition}\label{proposition-PKS}
Let $(\mathcal{M}, g_{ab})$ be a Lorentzian spacetime. 
Let $(o_A,\iota_A)$ be the dyad associated to $\mathcal{C}_{AA'}$ (cf. Remark \ref{remark-dyad}), and suppose that $o_A$ is SFR.
\begin{enumerate}
\item $o_A$ and $\iota_A$ are parallel spinors for \eqref{tildeC}:
\begin{equation}
 \tilde{\mathcal{C}}_{A'}o_B=0, \qquad \tilde{\mathcal{C}}_{A'}\iota_B=0. \label{PS}
\end{equation}
\item Suppose in addition that $(\mathcal{M}, g_{ab})$ is real-analytic and that $o^A$ is also a repeated principal spinor, 
so that eq. \eqref{f2} holds. Then the spinor field 
\begin{equation} \label{PKSnon-vacuum}
 K_{AB}=\phi^{-1}[o_{(A}\iota_{B)}-\eta o_Ao_B]
\end{equation}
satisfies the projected Killing spinor equation,
\begin{equation}
 o^{A}\nabla_{A'(A}K_{BC)}=0.
\end{equation}
\end{enumerate}
\end{proposition}

\begin{proof}
Equation \eqref{PS} follows simply from \eqref{Co}--\eqref{Ci}, so we see that analyticity is not required for this item.
For the second item, applying the definition \eqref{C} to the field $o_B\iota_C$, whose weights are $\{0,0\}$ and $w=1$, we have
\begin{equation}
 \mathcal{C}_{A'A}(o_B\iota_C)=\c_{A'A}(o_B\iota_C)+f_{A'A}o_B\iota_C-f_{A'B}o_A\iota_C-f_{A'C}o_B\iota_A. \label{coi0}
\end{equation}
But from \eqref{Co}--\eqref{Ci} we see that
\begin{equation}
 \mathcal{C}_{A'A}(o_B\iota_C)=\sigma'_{A'}o_Ao_Bo_C, \label{coi}
\end{equation}
thus $\mathcal{C}_{A'A}(o_B\iota_C)=\mathcal{C}_{A'(A}(o_B\iota_{C)})$. Combining with \eqref{coi0}, we get
\begin{equation}
 \mathcal{C}_{A'A}(o_B\iota_C)=\mathcal{C}_{A'(A}(o_B\iota_{C)})=(\nabla_{A'(A}-f_{A'(A})o_B\iota_{C)}.
\end{equation}
Using \eqref{f2}, this is
\begin{align*}
 \mathcal{C}_{A'A}(o_B\iota_C)={}& \phi\nabla_{A'(A}[\phi^{-1}o_B\iota_{C)}]+o_{(A}o_{B}\iota_{C)}\tilde{\mathcal{C}}_{A'}\eta \\
 ={}& \phi \left[ \nabla_{A'(A}[\phi^{-1}o_B\iota_{C)}]+\iota_{(A}\tilde{\mathcal{C}}_{|A'|}(\phi^{-1}\eta o_{B}o_{C)}) \right],
\end{align*}
where in the second line we used that $\tilde{\mathcal{C}}_{A'}o_B=0$, $\tilde{\mathcal{C}}_{A'}\phi=0$.
Contracting this equation with $o^A$ and using \eqref{coi} gives
\begin{equation}
 0=o^A\nabla_{A'(A}[\phi^{-1}o_B\iota_{C)}]+o^A\iota_{(A}\tilde{\mathcal{C}}_{|A'|}(\phi^{-1}\eta o_{B}o_{C)}).
\end{equation}
Let us compute the second term in the right hand side.
Using $o^A\tilde{\mathcal{C}}_{A'}(\cdot)=\tilde{\mathcal{C}}_{A'}(o^A\cdot)$, we get
\begin{equation*}
 o^A\iota_{(A}\tilde{\mathcal{C}}_{|A'|}(\phi^{-1}\eta o_{B}o_{C)})=-\tfrac{1}{3}\tilde{\mathcal{\C}}_{A'}(\phi^{-1}\eta o_Bo_C)
 =-o^A\mathcal{C}_{A'(A}[\phi^{-1}\eta o_{B}o_{C)}].
\end{equation*}
Noticing that the weights of the field $\phi^{-1}\eta o_{B}o_{C}$ are $\{0,0\}$ and $w=2$, and using the definition \eqref{C}, 
we have
\begin{equation}
 \mathcal{C}_{A'(A}[\phi^{-1}\eta o_{B}o_{C)}]=\nabla_{A'(A}[\phi^{-1}\eta o_{B}o_{C)}].
\end{equation}
Therefore
\begin{equation*}
 o^A\iota_{(A}\tilde{\mathcal{C}}_{|A'|}(\phi^{-1}\eta o_{B}o_{C)})=-o^A\nabla_{A'(A}(\phi^{-1}\eta o_{B}o_{C)})
\end{equation*}
and thus
\begin{equation}
0=o^A\nabla_{A'(A}\left[\phi^{-1}o_B\iota_{C)}-\phi^{-1}\eta o_{B}o_{C)}\right],
\end{equation}
hence the result follows.
\end{proof}

\begin{remark}
Any spinor field of the form $\w_{A}=f o_A$, where $f$ is of type $\{-1,0\}$, $w=0$, 
and satisfies $\tilde{\mathcal{C}}_{A'} f=0$, is a solution to the ``projected twistor equation''
\begin{equation}\label{ProjTE}
 o^A\nabla_{A'(A}\w_{B)}=0.
\end{equation}
This follows from the identity 
$o^A\mathcal{\C}_{A'(A}\w_{B)}=o^A\nabla_{A'(A}\w_{B)}$ replacing $\w_A=f o_A$.
\end{remark}

From Proposition \ref{proposition-PKS} we have that the spinor field \eqref{PKSnon-vacuum} is 
a projected Killing spinor in the general class of spacetimes where $o_A$ is SFR and a repeated PND but not necessarily Einstein, although in the non-Einstein case it must be real-analytic. This includes, in particular, the Kerr-(A)dS and (analytic) Kerr-Newman-(A)dS spacetimes, where $K_{AB}$ satisfies the usual, i.e. non-projected, Killing spinor equation.

\begin{example} \label{ex:KvFromKs}
Suppose that $(\mathcal{M}, g_{ab})$ is Einstein, i.e. $g_{ab}$ satisfies the vacuum Einstein equations with 
cosmological constant. Then, using \eqref{phieta}, we find that \eqref{PKSnon-vacuum} reduces to the 
expression \eqref{eq:ProjKs} given in the introduction.
Furthermore, the vector field 
\begin{align} \label{eq:xiDef2}
\xi_{AA'} = \nabla^B{}_{A'} K_{AB}
\end{align}
solve the projected Killing equation
\begin{align} 
o^A (\nabla_{AA'}\xi_{BB'}+\nabla_{BB'}\xi_{AA'}) = 0.
\end{align}
In GHP notation, \eqref{eq:xiDef2} takes the form
\begin{align} \label{eq:xiGHP}
\xi_a ={}&
 \frac{3}{2\Psi_2^{1/3}} \left( \rho' l_a - \rho n_a - \tau' m_a + \tau \bar{m}_a \right)
 + \frac{(\edt - 4\tau) \Psi_{3}}{2 \Psi_{2}^{4/3}} l_a
 -  \frac{(\tho - 4\rho) \Psi_{3}}{2 \Psi_{2}^{4/3}} m_a,
\end{align}
and from the Bianchi identities it follows that
\begin{align}
\xi^a \nabla_a \Psi_2 = 0.
\end{align} 
\end{example}
To give an explicit example, we review the Robinson-Trautman geometries in appendix~\ref{app:RobinsonTrautman}. There are known solutions of various Petrov types in this class and in particular there is one of Petrov type II.
\begin{example} \label{ex:xiForRT}
For the type II Robinson-Trautman solution given in \eqref{eq:TypeIIRTsol}, the vector \eqref{eq:xiGHP} is of the form
\begin{align}
\xi^a = \frac{-2 (-1)^{2/3}}{3 m^{4/3}} \left( (\zeta + \bar{\zeta})^3 (\partial_\zeta)^a + m (\partial_u)^a \right).
\end{align}
We note that it does not reduce to $\partial_u$, which is a Killing vector of that solution. 
\end{example}

Let us now discuss the space of projected Killing spinors. 
In type D spacetimes one can show that the space of valence--2 Killing spinors \eqref{KStypeD} is 1-dimensional, 
see e.g. \cite{jeffryes1984}.
For type II spacetimes, we first show the following:

\begin{proposition}
Let $(\mathcal{M}, g_{ab})$ be a Lorentzian spacetime, where $o^A$ is SFR and 
a repeated principal spinor. Let
\begin{align}\label{K}
K_{AB}={}& K_{2} o_{A} o_{B} - 2 K_{1} o_{(A}\iota_{B)} + K_{0} \iota_{A} \iota_{B}
\end{align}
be a projected Killing spinor, i.e. solving \eqref{PKSeq}. Then the components satisfy
\begin{subequations}
\begin{align}
  K_0 ={}& 0, \\
  \tilde{\mathcal{C}}_{A'} K_1 ={}& 0, \label{pks1} \\
  \tilde{\mathcal{C}}_{A'}K_2+2\mathcal{C}_{A'}K_1 ={}& 0. \label{pks2}
\end{align}
\end{subequations}
\end{proposition}

\begin{proof}
We first derive integrability conditions for \eqref{PKSeq} by introducing
\begin{align}
Q_{BCA'} =  o^A\nabla_{A'(A}K_{BC)}.
\end{align}
For any geometry satisfying \eqref{SFR} we find
\begin{align} \label{eq:KSIntCond}
\Theta_{(A}{}^{A'} Q_{BC)A'} = 
- 4 (\Psi_3 K_0 - \Psi_1 K_2) o_{(A} o_B \iota_{C)}
+ 4 (\Psi_2 K_0 - \tfrac{2}{3} \Psi_1 K_1) o_{(A}\iota_B \iota_{C)} 
- 2 \Psi_1 K_0 \iota_A \iota_B \iota_C .
\end{align}
The integrability conditions are then obtained by setting $Q_{BCA'}=0$, so the left hand side of \eqref{eq:KSIntCond} vanishes.
If $\Psi_1 \neq 0$, it follows that $K_{AB} = 0$ is the only solution to \eqref{PKSeq}, so we need to have at least $\Psi_1 = 0$ 
to have non-trivial solutions. 
Since we are interested in type II, we can impose $\Psi_1 = 0$.
In that case \eqref{eq:KSIntCond} leads to the restrictions
\begin{align}
\Psi_3 K_0 = 0, && \Psi_2 K_0 = 0.
\end{align}
If $K_0 \neq 0$, the Weyl spinor would be of Petrov type N (i.e. only $\Psi_4 \neq 0$). In that case \eqref{eq:KSIntCond} 
yields no restrictions on $K_{AB}$. For Petrov type D or II, we have $\Psi_2 \neq 0$ which forces $K_0 = 0$.
For $K_1$ and $K_2$, a compact form of the equations to be satisfied can be obtained in terms of the 
covariant derivative $\C_{AA'}$ and its projections $\tilde{\mathcal{C}}_{A'}$, $\mathcal{C}_{A'}$ given in \eqref{tildeC}. 
Using $K_0 = 0$, and requiring $K_{AB}$ to have conformal weight $2$ so that 
the projected Killing spinor equation is conformally invariant, eq. \eqref{PKSeq} is
\begin{equation}
 0=\tilde{\mathcal{C}}_{A'}K_{BC}+o_B\mathcal{C}_{A'C}K_1+o_C\mathcal{C}_{A'B}K_1.
\end{equation}
Contraction with $o^B$ and with $\iota^B\iota^C$ leads to,
\begin{subequations}
\begin{align}
  \tilde{\mathcal{C}}_{A'} K_1 ={}& 0, \\
  \tilde{\mathcal{C}}_{A'}K_2+2\mathcal{C}_{A'}K_1 ={}& 0,
\end{align}
\end{subequations}
or in GHP notation
\begin{subequations}
\begin{align}
(\edt + \tau) K_{1}={}&0,&&&
(\tho + \rho) K_{1} ={}&0,\\
2 (\tho' + \rho') K_{1} + (\edt + 2\tau) K_{2}={}&0, &&&
 2 (\edt' + \tau') K_{1} + (\tho + 2\rho) K_{2}={}&0.
\end{align}
\end{subequations}
\end{proof}

\begin{remark}[Solution space of the projected Killing spinor equation]
In the real-analytic case, we see from eqs. \eqref{pks1}-\eqref{pks2} that the space of solutions 
to the projected Killing spinor equation is infinite dimensional. 
For example, one solution is given in \eqref{PKSnon-vacuum}, but we also see that, taking $K_1\equiv 0$, equations 
\eqref{pks1}-\eqref{pks2} reduce to $\tilde{\mathcal{C}}_{A'}K_2=0$.
It follows from equations \eqref{identity-tildeC}, \eqref{constantbetas} and \eqref{constantbetas2} that there are infinitely many 
solutions to this equation.
Any function of the form $F=\l^{-2}\phi^2 K_2$, see Remark~\ref{remark-tildeC} and Lemma~\ref{lemma-tildeC} 
for the definition of $\lambda$ and $\phi$,
which is constant on $\beta$-surfaces leads to a projected Killing spinor $K_{AB}=\l^{2}\phi^{-2}F o_Ao_B$.
We also mention that one could in principle consider $K_1 \neq 0$ and 
$K_2 \equiv 0$ (or $\tilde{\mathcal{C}}_{A'}K_2=0$), but then \eqref{pks1}-\eqref{pks2} 
leads to $\mathcal{C}_{AA'}K_1=0$, which, taking into account that the weights of $K_1$ are 
$p=0$ and $w=1$, is explicitly $\c_{a}K_1+f_a K_1 =0$. This implies $f_a=-\c_a\log K_1$ 
and so $\c_{[a}f_{b]}=0$, which is a restriction on the geometry: 
this is satisfied for vacuum type D but not for arbitrary type II spacetimes.
\end{remark}

Finally, an interesting property of projected Killing spinors is that they give rise to solutions of the vacuum 
Maxwell equations, and therefore to conservation laws:
\begin{lemma}\label{lemma-maxwell}
Let $(\mathcal{M}, g_{ab})$ be real-analytic, and let $o^A$ be SFR and a repeated principal spinor.
Let $K_{AB}$ be a projected Killing spinor as in \eqref{K}. Assume $K_1 = c\phi^{-1}$, where $\phi$ was introduced in \eqref{f2}
and $c$ is an arbitrary constant, possibly zero.  Then the spinor field
\begin{equation}\label{lemma-nonvacuum}
  \varphi_{AB}=-2\phi^3 K_{AB}
\end{equation}
is a solution to the vacuum Maxwell equations.
\end{lemma}

The result of this Lemma is valid regardless of the Einstein condition, but it imposes the real-analyticity restriction, because we need the existence of the scalar field $\phi$ introduced in Lemma~\ref{lemma-f}.
If the Einstein condition is assumed, then real-analyticity is not needed, see Example~\ref{example-spinlowering} below.

\begin{proof}[Proof of Lemma~\ref{lemma-maxwell}]
The proof is straightforward in terms of the conformal connection \eqref{C}.
Since the field \eqref{lemma-nonvacuum} has GHP weight $\{0,0\}$ and conformal weight $w=-1$,
it follows from \eqref{C} that $\nabla_{A'}{}^{A}\varphi_{AB}=\mathcal{\C}_{A'}{}^{A}\varphi_{AB}$.
Replacing expression \eqref{K} for $K_{AB}$ with $K_0=0$, and using $\mathcal{\C}_{AA'}o_B=0$,
we get
\begin{align}
\nonumber -\tfrac{1}{2}\nabla_{A'}{}^{A}\varphi_{AB}={}& 
-o_B\tilde{\mathcal{C}}_{A'}(\phi^3K_2)+\tilde{\mathcal{C}}_{A'}(\phi^3 K_1\iota_B)-o_B\mathcal{C}_{A'}{}^{A}(\phi^3K_1\iota_A)\\
\nonumber ={}& -o_B\phi^3\tilde{\mathcal{C}}_{A'}K_2+o_B\mathcal{C}_{A'}(\phi^3 K_1) \\
\nonumber ={}& 2o_B\phi^3\mathcal{C}_{A'}K_1+o_B\mathcal{C}_{A'}(\phi^3 K_1) \\
\nonumber ={}& 3o_B\phi^3\mathcal{C}_{A'}K_1+o_BK_1\mathcal{C}_{A'}\phi^3 \\
 ={}& 3o_B\phi^2\mathcal{C}_{A'}(\phi K_1),
\end{align}
where in the second line we used $0=\tilde{\mathcal{C}}_{A'}\phi=\tilde{\mathcal{C}}_{A'}K_1=\tilde{\mathcal{C}}_{A'}\iota_B$ 
and also $\mathcal{\C}_{A'}{}^{A}\iota_A=0$ (which follows from \eqref{Ci}), and in the third line we used eq. \eqref{pks2}.
Now, the scalar $\phi K_1$ has weights $\{0,0\}$ and $w=0$, so from the definition \eqref{C} we have 
$\mathcal{\C}_{A'}(\phi K_1)=\iota^A\nabla_{AA'}(\phi K_1)$.
Therefore $\nabla_{A'}{}^{A}\varphi_{AB}=0$ if and only if $\iota^A\nabla_{AA'}(\phi K_1)=0$. Since we are assuming 
a generic type II spacetime with SFR $o^A$, this is true if and only if $K_1=c\phi^{-1}$ for some constant $c$, 
which can be zero.
This follows from the discussion in Remark \ref{remark-iota}.
\end{proof}

The result of Lemma \ref{lemma-maxwell} tells us that, as long as $K_1=c\phi^{-1}$, 
we have the conservation law
\begin{equation}\label{conservationlaw}
 {\rm d}\mc{F}=0, \qquad \mc{F}_{ab}=\varphi_{AB}\e_{A'B'}.
\end{equation}
We note that, choosing $c=0$, the Maxwell field \eqref{lemma-nonvacuum} is $\varphi_{AB}=-2\phi^3K_2 o_Ao_B$, 
with $\tilde{\mathcal{\C}}_{A'}(\phi^3K_2)=0$. These are the \textit{Robinson null fields} discussed for example in \citep[Theorem 7.3.14]{PR2}.

\begin{example}[Spin lowering]\label{example-spinlowering}
If we assume that $(\mathcal{M}, g_{ab})$ is Einstein, then from \eqref{phieta} we know that $\phi=\Psi^{1/3}_2$ 
and we do not need to restrict to real-analytic spacetimes. 
The Maxwell field \eqref{lemma-nonvacuum} is $\varphi_{AB}=-2\Psi_2 K_{AB}$.
Choosing $c=1$, it follows that
\begin{equation}\label{maxwell-vacuum}
 \varphi_{AB}=\Psi_{ABCD}K^{CD}.
\end{equation}
Thus, in this case, the result of Lemma \ref{lemma-maxwell} can be interpreted as a form of Penrose's spin lowering:
\begin{equation}
 \nabla^{AA'}\varphi_{AB}=(\nabla^{AA'}\Psi_{ABCD})K^{CD}+\Psi_{ABCD}(\nabla^{A'(A}K^{CD)})=0
\end{equation}
where the first term on the right vanishes since $\nabla^{AA'}\Psi_{ABCD}=0$ are the Bianchi identities for an Einstein space, 
and the vanishing of the second term follows from the fact that the projected Killing spinor equation \eqref{PKSeq}
implies $\nabla^{A'(A}K^{CD)}=A^{A'}o^{A}o^{C}o^{D}$ for some $A^{A'}$, and then one uses the type II condition \eqref{typeII}.
\end{example}

\begin{example}[Sachs' conservation law]
In \cite[eq.(5.23)]{1961Sachs}, Sachs found a conservation law for vacuum type II spacetimes, where he 
chose a tetrad rotation such that $\Psi_3 = 0$ to derive the conservation law
\begin{align}
 {\rm d} \left( \Psi_2^{2/3} o_{(A}\iota_{B)} \bar{\epsilon}_{A'B'} \right) = 0.
\end{align}
From Example \ref{example-spinlowering} we see that this conservation law can be interpreted as a form of 
Penrose's spin lowering without the necessity of imposing the tetrad gauge $\Psi_3 = 0$, since 
the components of the Maxwell field \eqref{maxwell-vacuum} are 
\begin{align}
\varphi_{0}={}0, &&
\varphi_{1}={}\Psi_{2}^{2/3},&&
\varphi_{2}={}\frac{2 \Psi_{3}}{3 \Psi_{2}^{1/3}},
\end{align}
which, using \eqref{conservationlaw}, shows the equivalence to the result of Sachs.
\end{example}

\section{Perturbation theory in terms of a Debye potential}\label{sec:Debye}

In this section we review first order perturbation theory of Einstein spacetimes of Petrov type II in terms of scalar potentials and prove Theorem~\ref{thm}. We make use of the adjoint operator method introduced by Wald in \citep{Wald} to generate solutions to the linearized Einstein equations from solutions of a scalar wave-like equation.

First, recall the expression for the linearized ASD Weyl spinor $\dot\Psi_{ABCD}$ in terms of the linearized metric for an Einstein background with metric $g_{ab}$, as given by Penrose and Rindler in \cite[Eq. (5.7.15)]{PR1},
\begin{equation}\label{linearizedWeyl}
 \dot\Psi_{ABCD}=\tfrac{1}{2} \nabla_{(A}{}^{A'}\nabla_{B}{}^{B'} h_{CD)A'B'}  + \tfrac{1}{4} g^{ef}h_{ef} \Psi_{ABCD}.
\end{equation}
\begin{remark}
We note that the spinor variational operator $\vartheta$ introduced in \cite{2016:Backdahl:ValienteKroon} leads to a minus sign of the trace term in \eqref{linearizedWeyl}. Since the linearized metrics generated from a scalar potential are trace-free, i.e. $g^{ef}h_{ef}=0$, see \eqref{eq:MetricFromDeybe} for details, the result would be the same for $\vartheta\Psi_{ABCD}$. However, since we are only interested in the ASD Weyl curvature here, there is no need to introduce the additional structures involving the $\vartheta$ variation.
\end{remark}
To state the underlying operator identity of Wald's method, we begin with
\begin{definition}\label{def:AdjOp}
Let $x_{ab}$ be a (possibly complex) symmetric 2-tensor, $\Phi$ be a complex scalar field of GHP weight $\{4,0\}$ and  denote $l^{a}l^{b}x_{ab}$ by $x_{ll}$ etc. for tetrad components. Define the differential operators $\mathcal{S}, \mathcal{E}, \mathcal{O}, \mathcal{T}$ by
\begin{subequations}
\begin{align}
\mathcal{S}(x_{ab}) :={}& (\edt-4\tau-\bar{\tau}')\left((\tho-2\bar{\rho})x_{lm}-(\edt-\bar{\tau}')x_{ll}\right) \nonumber \\ 
{}& +(\tho-4\rho-\bar{\rho})\left(\bar{\sigma}'x_{ll}+(\edt-2\bar{\tau}')x_{lm}-(\tho-\bar{\rho})x_{mm}\right), \label{eq:operatorS}\\
 \mathcal{E}(x)_{ab} :={}& (-\tfrac{1}{2}\Box + 6\Lambda) x_{ab} -\tfrac{1}{2}\nabla_{a}\c_{b}(g^{cd}x_{cd}) +\nabla^{c}\nabla_{(a}x_{b)c}
 -\tfrac{1}{2}g_{ab}(\nabla^{c}\c^{d}x_{cd}-\Box (g^{cd}x_{cd}) ), \label{eq:operatorE} \\ 
 \mathcal{O}(\Phi) :={}& 2[(\tho-4\rho-\bar{\rho})(\tho'-\rho')-(\edt-4\tau-\bar{\tau}')(\edt'-\tau')-3\Psi_2]\Phi, \label{eq:operatorO}\\
    \mathcal{T}(x_{ab})  :={}&
       \tfrac{1}{2} \left((\edt - 2 \bar{\tau}')\edt - \bar{\sigma}' (\tho - 2 \bar{\rho}) - 2(\tho \bar{\sigma}')\right) x_{ll}  + \tfrac{1}{2} (\tho - 2 \bar{\rho})\tho x_{mm} \nonumber \\
 {}& -  \left((\tho - 2 \bar{\rho})(\edt -  \bar{\tau}') - (\edt \bar{\rho}) \right)x_{lm}.   \label{eq:operatorT}
\end{align}
\end{subequations}
\end{definition}

\begin{remark}
\begin{enumerate}
 \item Note that $\mathcal{E}$ defined in \eqref{eq:operatorE} is the linearized Einstein operator plus a cosmological term. This operator is formally self adjoint,
 \begin{align} \label{eq:Eselfadjointness}
 \mathcal{E}^\dagger = \mathcal{E}.
 \end{align}
 \item  Acting on a linearized metric $h_{ab}$, the operator $\mathcal{T}$ defined in \eqref{eq:operatorT} yields the 0-component of the linearized anti-self dual Weyl curvature \eqref{linearizedWeyl},
\begin{align}
\mathcal{T}(h_{ab}) ={}& o^{A}o^{B}o^{C}o^{D} \dot\Psi_{ABCD} =  \dot\Psi_0.
\end{align}
 \item For any $h_{ab}$ such that $\mathcal{E}(h)_{ab}=0$, a decoupled wave-like (Teukolsky) equation is given by
 \begin{align} \label{eq:Psi0TME}
  \mathcal{O}\mathcal{T}(h_{ab}) = \mathcal{O}(\dot\Psi_0) = 0.
 \end{align}
\end{enumerate}
\end{remark}

\begin{theorem}[Wald \cite{Wald}]
On Einstein spacetimes of Petrov type II with repeated principal spinor $o^A$, the operators of definition~\ref{def:AdjOp} satisfy the identity
\begin{equation}\label{SEOT}
 \mathcal{S}\mathcal{E}=\mathcal{O}\mathcal{T}.
\end{equation}
Thus, in particular, if $\chi$ is a complex scalar field of GHP weight $\{-4,0\}$ solving
\begin{align} \label{eq:DebyeEquation1}
\mathcal{O}^{\dag} \chi=0,
\end{align} 
then the complex metric $h_{ab} = \mathcal{S}^{\dag}(\chi)_{ab}$ solves the linearized Einstein equation
\begin{align}
\mathcal{E}(h)_{ab} =\mathcal{E}(\mathcal{S}^{\dag}(\chi))_{ab} = 0.
\end{align}
\end{theorem}

In \cite{Wald} the general idea was outlined and applied for Petrov type D, see \cite{2018:Prabhu:Wald}\footnote{It seems that there is a factor 4 missing in the operator identity in that reference.} where the explicit form of $\mathcal{S}$ is given on Petrov type II. The scalar field $\chi$ is called Debye potential and \eqref{eq:DebyeEquation1} is the Debye equation.

\begin{lemma}
Let $\chi$ be a complex scalar field of GHP weight $\{-4,0\}$.
\begin{itemize}
 \item The operators adjoint to $\mathcal{S}$ and $\mathcal{O}$ are given by
\begin{subequations} \label{eq:SOadj}
\begin{align}
\mathcal{S}^\dagger(\chi)_{ab} ={}&  \left( l_{(a}m_{b)}(\tho - \rho + \bar\rho) -l_a l_b (\edt-\tau) \right)(\edt+3\tau)\chi \nonumber \\
{}& + \left(l_{(a}m_{b)}(\edt-\tau+\bar\tau') - l_a l_b \bar\sigma' - m_a m_b (\tho - \rho)  \right)(\tho + 3\rho)\chi, \label{eq:Sadj}\\
\mathcal{O}^{\dag}\chi ={}& 2\left(\left( \tho' - \bar{\rho}' \right) \left(\tho + 3 \rho \right) - \left(\edt' - \bar{\tau}\right) \left(\edt + 3 \tau \right)  - 3 \Psi_2 \right) \chi. \label{eq:Oadj}
\end{align}
\end{subequations}
\item If $\chi$ solves the Debye equation \eqref{eq:DebyeEquation1}, then the complex metric $h_{ab} = \mathcal{S}^{\dag}(\chi)_{ab}$ solving the linearized Einstein equation is given by
\begin{align}\label{eq:MetricFromDeybe}
h_{ab} = {}& h_{nn} l_a l_b - 2 h_{n\bar{m}} l_{(a} m_{b)} + h_{\bar{m}\bar{m}} m_a m_b, 
\end{align}
with components
\begin{subequations} \label{eq:MetricFromDeybeComps}
\begin{align}
h_{nn}={}& -(\edt+2\tau) \edt \chi - \bar{\sigma}' \tho \chi,\\
h_{n\bar{m}}={}& -(\tho +\rho ) \edt \chi - (\tau + \bar{\tau}')\tho\chi,\\
h_{\bar{m}\bar{m}}={}& -(\tho+2\rho)\tho\chi.
\end{align}
\end{subequations} 
\end{itemize}
\end{lemma}
\begin{proof}
The formal adjoints of the GHP operators are given by
\begin{align}
\tho^\dagger ={}& - \tho + \rho + \bar\rho, &&&
\tho'^\dagger ={}& - \tho' + \rho' + \bar\rho', &&&
\edt^\dagger ={}& - \edt + \tau + \bar\tau', &&&
\edt'^\dagger ={}& - \edt' + \tau' + \bar\tau,
\end{align} 
see e.g. \cite{aksteiner:thesis} for details. Using this, together with the rule $(AB)^\dagger = B^\dagger A^\dagger$ for compositions, to compute the adjoint of \eqref{eq:operatorS} and \eqref{eq:operatorO} leads to \eqref{eq:SOadj}.

The equations \eqref{eq:MetricFromDeybeComps} are the simplified components of \eqref{eq:Sadj}, using the commutator 
\begin{align}
[\edt,\tho]\chi ={}& (\bar\tau'\tho -\bar\rho\edt)\chi, 
\end{align}
and Ricci identities,
\begin{align}
\tho \rho ={}& \rho^2, &&&
\edt \tau ={}& \tau^2 - \bar\sigma'\rho, &&&
\tho\tau ={}& \rho(\tau-\bar\tau'), &&&
\edt\rho ={}& \tau(\rho-\bar\rho),
\end{align}
valid on type II Einstein spacetimes with repeated principal spinor $o^A$.
\end{proof}

Before proceeding to the proof of Theorem \ref{thm}, we need to introduce some additional identities and operators. 
First, similarly to Definition \ref{Definition-projectedC}, we introduce: 
\begin{definition}
 Let $(o^A,\iota^A)$ be a spin dyad, $\Theta_{AA'}$ the associated GHP connection, and $f_{AA'}$ the 1-form \eqref{f}. We define
 \begin{align}
  \tilde{\Theta}_{A'}:=&o^{A}\Theta_{AA'}, & \Theta_{A'}:=&\iota^{A}\Theta_{AA'}, \label{projectedTheta} \\
  \tilde{f}_{A'}:=&o^{A}f_{AA'}, & f_{A'}:=&\iota^{A}f_{AA'}, \label{projectedf} \\
  && \sigma'_{A'}:=&\iota^A\iota^B\c_{AA'}\iota_B. \label{sigma'A'}
 \end{align}
\end{definition}

\begin{proposition}
If $o^A$ is SFR, then
\begin{align} \label{eq:FirstStructureEq}
  \tilde{\T}_{A'}o^B = 0,  \qquad
  \T_{A'}o^B = -\tilde{f}_{A'}\iota^B, \qquad
  \tilde{\T}_{A'}\iota^B = -f_{A'}o^B, \qquad
  \T_{A'}\iota^B = \sigma'_{A'} o^B.
\end{align}
\end{proposition}

\begin{remark}
If $o^A$ is SFR and a repeated principal spinor, using \eqref{f2} we note that $\tilde{f}_{A'}$ can be written as
\begin{equation}\label{tildef}
 \tilde{f}_{A'}=\tilde{\Theta}_{A'}\log\phi.
\end{equation}
Since in this section we are interested in the Einstein case, we can use \eqref{phieta} and replace $\phi=\Psi^{1/3}_2$.
\end{remark}

\begin{proposition}
Let $(\mathcal{M}, g_{ab})$ be a type II Einstein spacetime with repeated principal spinor $o^A$. On spinors of GHP weight $\{p,0\}$ with an arbitrary number of primed indices the operators in \eqref{projectedTheta} satisfy the following commutator relations,
\begin{subequations} \label{eq:ProjCommutators}
\begin{align}
  [\tT_{A'},\tT_{B'}] ={}& 0, \label{commutator-tildeTheta} \\
 [\tT_{A'},\T_{B'}] ={}& \e_{A'B'}\frac{p}{2}(2\tT_{C'}f^{C'}-\tilde f_{C'}f^{C'} +  6\Lambda)+\Box^{\T}_{A'B'}-f_{A'}\tT_{B'}+\tilde{f}_{B'}\T_{A'}, \label{commutator-tildeThetaTheta}\\
 [\T_{A'},\T_{B'}] ={}& \e_{A'B'}\left(\frac{p}{2}(\Theta_{C'}f^{C'} - \tT_{C'}\sigma'^{C'} - \tilde f_{C'}\sigma'^{C'})-\sigma'^{C'}\tT_{C'}\right). \label{commutator-Theta}
\end{align}
\end{subequations}
\end{proposition}

\begin{proof}
We prove \eqref{commutator-tildeTheta} explicitly, the other equations follow analogously. In a general spacetime, using definition \eqref{projectedTheta} and 
acting on an arbitrary spinor $\varphi^{E...E'...}_{F...F'...}$ with GHP weight $\{p,q\}$ we get
\begin{equation}
 [\tilde{\Theta}_{A'},\tilde{\Theta}_{B'}]\varphi^{E...E'...}_{F...F'...}
 =\bar\epsilon_{A'B'}(o^Ao^B\nabla_{A}{}^{C'}o_{B})\iota^{C}\Theta_{CC'}\varphi^{E...E'...}_{F...F'...}
 +o^{A}o^{B}[\Theta_{AA'},\Theta_{BB'}]\varphi^{E...E'...}_{F...F'...}. \label{commutator-tildeTheta-general}
\end{equation}
For the second term in the right, using GHP notation we find
\begin{align}
\nonumber o^{A}o^{B}[\Theta_{AA'},\Theta_{BB'}]\varphi^{E...E'...}_{F...F'...}
 ={}& \bar\epsilon_{A'B'} o^{A}o^{B}\Theta_{AC'}\Theta_{B}{}^{C'}\varphi^{E...E'...}_{F...F'...}  \\
\nonumber ={}& \bar\epsilon_{A'B'} o^{A}o^{B}\Box_{AB}\varphi^{E...E'...}_{F...F'...} \\
 & +\left( p(-\Psi_1+\sigma\tau'-\kappa\rho')+q(-\Phi_{01}+\bar\tau'\bar\rho-\bar\kappa\bar\sigma') \right)\varphi^{E...E'...}_{F...F'...} 
 \label{oocommutator}
\end{align}
where $\Box_{AB}=\nabla_{A'(A}\nabla_{B)}{}^{A'}$ is the usual spinor curvature operator \cite[Eq. (4.9.2)]{PR1}.
Suppose now that the spacetime is Einstein and of Petrov type II. 
Restricting the identity \eqref{oocommutator} to {\em primed} spinor fields $\varphi^{E'...}_{F'...}$ with GHP weight $\{p,0\}$, 
we get that each of the terms in the right hand side vanishes: $\Box_{AB}\varphi^{E'...}_{F'...}=0$ because this only involves 
contractions with $\Phi_{ABA'B'}$, which vanishes because of the Einstein condition; the term with $q$ vanishes because 
we are restricting to weight $\{p,0\}$; and the term with $p$ vanishes because the vacuum type II condition implies 
$\kappa=\sigma=\Psi_1=0$. 
Furthermore, from this last condition we get $o^Ao^B\nabla_{A}{}^{C'}o_{B}=0$ (see \eqref{SFR}), therefore 
\eqref{commutator-tildeTheta-general} vanishes identically and we get the result \eqref{commutator-tildeTheta}.
\end{proof}

\begin{remark}
Since $[\tilde{\Theta}_{A'},\tilde{\Theta}_{B'}]=2\tilde{\Theta}_{[A'}\tilde{\Theta}_{B']}=-\bar{\epsilon}_{A'B'}\tilde{\Theta}^{C'}\tilde{\Theta}_{C'}$,
the result \eqref{commutator-tildeTheta} can be equivalently stated as
\begin{equation}\label{squareTildeTheta}
 \tilde{\Theta}^{A'}\tilde{\Theta}_{A'} \varphi^{E'...}_{F'...} = 0 
\end{equation}
for any primed spinor $\varphi^{E'...}_{F'...}$ (or scalar) with weight $\{p,0\}$. This identity will be useful below.
\end{remark}

\begin{proposition}
Let $\tilde{\Theta}_{A'}$ be as in \eqref{projectedTheta}, and let $\phi$ be given by \eqref{phieta}.
The linearized metric \eqref{eq:MetricFromDeybe} can be expressed as 
\begin{equation}\label{ooX}
 h_{AA'BB'}=o_Ao_B X_{A'B'},
\end{equation}
where 
\begin{align} 
 X_{A'B'}
 ={}& -\phi^{-2} \tilde{\Theta}_{(A'}\left( \phi^{2} \tilde{\Theta}_{B')}\chi\right).\label{X}
\end{align}
\end{proposition}

\begin{proof}
This can be shown by direct comparison with \eqref{eq:MetricFromDeybe}-\eqref{eq:MetricFromDeybeComps}.
First, from \eqref{ooX} we see that the only non-trivial components are 
\begin{subequations}\label{componentsX}
\begin{align}
 X_{1'1'} ={}&- \iota^{A}\bar{\iota}^{A'}\iota^{B}\bar{\iota}^{B'}h_{AA'BB'}=-n^{a}n^{b}h_{ab}= -h_{nn}, \\
 X_{0'1'} ={}& -\iota^{A}\bar{\iota}^{A'}\iota^{B}\bar{o}^{B'}h_{AA'BB'}=-n^{a}\bar{m}^{b}h_{ab}=- h_{n\bar{m}}, \\
 X_{0'0'} ={}& -\iota^{A}\bar{o}^{A'}\iota^{B}\bar{o}^{B'}h_{AA'BB'}=-\bar{m}^{a}\bar{m}^{b}h_{ab}=- h_{\bar{m}\bar{m}}, 
\end{align}
\end{subequations}
so the general structure \eqref{eq:MetricFromDeybe} is recovered.
Now we check that these components coincide with \eqref{eq:MetricFromDeybeComps}. 
To do this, we write $X_{A'B'}$ as
\begin{equation}
 X_{A'B'}=-\tilde{\Theta}_{A'}\tilde{\Theta}_{B'}\chi - 2(\tilde{\Theta}_{(A'}\log\phi)(\tilde{\Theta}_{B')}\chi).
\end{equation}
Next, we project this expression over the primed spin dyad $(\bar{o}^{A'},\bar{\iota}^{A'})$ to compute the 
components \eqref{componentsX}.
We will need the following identities:
\begin{equation}
 \tho\log\phi=\rho, \qquad \edt\log\phi=\tau,
\end{equation}
which follow from the fact that $\phi=\Psi^{1/3}_2$ (eq. \eqref{phieta}) together with the Bianchi identities for a Petrov type II spacetime. Then we have, for example,
\begin{align}
\nonumber X_{1'1'} ={}& -\bar\iota^{A'}\bar\iota^{B'}\tilde{\Theta}_{A'}\tilde{\Theta}_{B'}\chi - 2\edt\log\phi\edt\chi \\
\nonumber ={}& -\edt\edt\chi+(\edt\bar\iota^{B'})\tilde{\Theta}_{B'}\chi - 2\tau\edt\chi \\
 ={}& -(\edt+2\tau)\edt\chi - \bar\sigma'\tho\chi
\end{align}
where we used $\edt\bar\iota^{B'}=-\bar\sigma'\bar{o}^{B'}$, see \cite[Eq. (4.12.28)]{PR1}.
The other components can be computed along the same line.
\end{proof}

\begin{proposition}
Let $(\mathcal{M}, g_{ab})$ be a type II Einstein spacetime with repeated principal spinor $o^A$. The Ricci, Bianchi and commutator identities yield
\begin{subequations} \label{eq:ThThtSimpSet}
\begin{align}
\tilde\Theta_{A'} \tilde f_{B'} ={}& \tilde f_{A'} \tilde f_{B'}, \\
\tilde\Theta^{A'}f_{A'} ={}& \Psi_2 + 2 \Lambda, \\
\Theta^{A'}\tilde f_{A'} ={}& -\Psi_2 - 2 \Lambda, \\
\tilde\Theta_{A'} \Psi_2 ={}& 3 \tilde f_{A'} \Psi_2, \\
 \tilde{\Theta}_{A'}\tilde{\Theta}_{B'}\phi^{-1} ={}& 0, \label{squareThetaphi}
\end{align}
\end{subequations}
with the scalar field $\phi$ defined in \eqref{phieta}.
\end{proposition}
\begin{proof}
These equations can be checked by GHP expansion using for example \cite[Section 4.12]{PR1}.
\end{proof}
The main result of this section is given in lemma~\ref{lem:asdDebyeCurvature} below, from which theorem~\ref{thm} follows. The proof of lemma~\ref{lem:asdDebyeCurvature} involves long computations and we transfer some intermediate steps into the following:
\begin{proposition}
The Debye equation in this formulation is given by
\begin{align} \label{eq:ThThtDebye}
\mathcal{O}^\dagger \chi = 2 \Theta^{A'} (\tilde\Theta_{A'} + 3\tilde f_{A'} )\chi - 6\Psi_2 \chi.
\end{align}
Repeated application of commutators on $\chi$ of weight $\{-4,0\}$ leads to
\begin{align}\label{eq:Psi3CommEq1}
\tilde{\Theta}_{A'}\Theta_{B'}\tilde{\Theta}^{B'}\tilde{\Theta}^{A'}\chi ={}&- \tilde{f}^{A'} \tilde{f}^{B'} \Theta_{B'}\tilde{\Theta}_{A'}\chi
 + 3 \tilde{f}^{A'} \Lambda \tilde{\Theta}_{A'}\chi
 -  \tfrac{15}{2} \tilde{f}^{A'} \Psi_{2} \tilde{\Theta}_{A'}\chi
 -  \tfrac{3}{2} f^{A'} \tilde{f}^{B'} \tilde{\Theta}_{A'}\tilde{\Theta}_{B'}\chi\nonumber\\
& - 3 f^{A'} \tilde{f}_{A'} \tilde{f}^{B'} \tilde{\Theta}_{B'}\chi
 + \tilde{f}^{A'} \tilde{\Theta}_{B'}\Theta^{B'}\tilde{\Theta}_{A'}\chi
 -  \tfrac{3}{2} f^{A'} \tilde{f}^{B'} \tilde{\Theta}_{B'}\tilde{\Theta}_{A'}\chi\nonumber\\
& + \tfrac{3}{2} \tilde{f}^{A'} \tilde{\Theta}_{A'}f_{B'} \tilde{\Theta}^{B'}\chi
 -  \tfrac{3}{2} \tilde{f}^{A'} \tilde{\Theta}_{B'}f_{A'} \tilde{\Theta}^{B'}\chi .
\end{align}
\begin{align}\label{eq:Psi3CommEq2}
\tilde{\Theta}_{B'}\Theta_{A'}\tilde{\Theta}^{B'}\tilde{\Theta}^{A'}\chi ={}&- \tilde{f}^{A'} \tilde{f}^{B'} \Theta_{B'}\tilde{\Theta}_{A'}\chi
 + 3 \tilde{f}^{A'} \Lambda \tilde{\Theta}_{A'}\chi
 -  \tfrac{15}{2} \tilde{f}^{A'} \Psi_{2} \tilde{\Theta}_{A'}\chi
 -  \tfrac{3}{2} f^{A'} \tilde{f}^{B'} \tilde{\Theta}_{A'}\tilde{\Theta}_{B'}\chi\nonumber\\
& - 3 f^{A'} \tilde{f}_{A'} \tilde{f}^{B'} \tilde{\Theta}_{B'}\chi
 + \tilde{f}^{A'} \tilde{\Theta}_{B'}\Theta^{B'}\tilde{\Theta}_{A'}\chi
 -  \tfrac{3}{2} f^{A'} \tilde{f}^{B'} \tilde{\Theta}_{B'}\tilde{\Theta}_{A'}\chi\nonumber\\
& + \tfrac{3}{2} \tilde{f}^{A'} \tilde{\Theta}_{A'}f_{B'} \tilde{\Theta}^{B'}\chi
 -  \tfrac{3}{2} \tilde{f}^{A'} \tilde{\Theta}_{B'}f_{A'} \tilde{\Theta}^{B'}\chi .
\end{align}

\begin{align}\label{eq:Psi4CommEq1}
\Theta_{B'}\Theta_{A'}\tilde{\Theta}^{B'}\tilde{\Theta}^{A'}\chi ={}&-9 \chi f^{A'} \tilde{f}_{A'} \Lambda
 + 18 \chi \Lambda^2
 -  \tfrac{3}{4} f^{A'} \tilde{f}_{A'} \mathcal{O}^{\dagger}{}\chi
 + \tfrac{9}{2} \Lambda \mathcal{O}^{\dagger}{}\chi
 + \tfrac{1}{4} \mathcal{O}^{\dagger}{}\mathcal{O}^{\dagger}{}\chi
 - 9 \chi f^{A'} \tilde{f}_{A'} \Psi_{2}\nonumber\\
& - 36 \chi \Lambda \Psi_{2}
 + \tfrac{3}{2} \mathcal{O}^{\dagger}{}\chi \Psi_{2}
 - 36 \chi \Psi_{2}^2
 + 3 \tilde{f}^{A'} \Lambda \Theta_{A'}\chi
 - 30 \tilde{f}^{A'} \Psi_{2} \Theta_{A'}\chi\nonumber\\
& + \tfrac{3}{2} \tilde{f}^{A'} \Theta_{A'}\mathcal{O}^{\dagger}{}\chi
 - 18 \chi \tilde{f}^{A'} \Theta_{A'}\Psi_{2}
 - 3 \Lambda \Theta_{A'}\tilde{\Theta}^{A'}\chi
 - 6 \Psi_{2} \Theta_{A'}\tilde{\Theta}^{A'}\chi\nonumber\\
& -  \tfrac{3}{2} f^{A'} \tilde{f}^{B'} \Theta_{A'}\tilde{\Theta}_{B'}\chi
 -  \tfrac{9}{2} f^{A'} \tilde{f}_{A'} \tilde{f}^{B'} \Theta_{B'}\chi
 - 3 \tilde{f}^{A'} \tilde{f}^{B'} \Theta_{B'}\Theta_{A'}\chi\nonumber\\
& + \tilde{f}^{A'} \Theta_{B'}\Theta^{B'}\tilde{\Theta}_{A'}\chi
 -  \tfrac{3}{2} f^{A'} \tilde{f}^{B'} \Theta_{B'}\tilde{\Theta}_{A'}\chi
 - 3 \tilde{f}^{A'} \Theta_{B'}\tilde{\Theta}^{B'}\Theta_{A'}\chi\nonumber\\
& + 3 \tilde{f}^{A'} \Theta_{A'}\tilde{f}_{B'} \Theta^{B'}\chi
 -  \Theta_{B'}\tilde{\Theta}_{A'}\chi \Theta^{B'}\tilde{f}^{A'}
 - 3 f^{A'} \Lambda \tilde{\Theta}_{A'}\chi
 -  \tfrac{3}{2} f^{A'} \Psi_{2} \tilde{\Theta}_{A'}\chi\nonumber\\
& + 3 \Theta^{A'}\Psi_{2} \tilde{\Theta}_{A'}\chi
 -  \tfrac{3}{2} \tilde{f}^{A'} \Theta_{B'}f^{B'} \tilde{\Theta}_{A'}\chi
 -  \tfrac{3}{2} \tilde{f}^{A'} \Theta_{A'}f^{B'} \tilde{\Theta}_{B'}\chi\nonumber\\
& -  \tfrac{3}{2} f^{A'} \Theta_{A'}\tilde{f}^{B'} \tilde{\Theta}_{B'}\chi
 + \tfrac{3}{2} \tilde{f}^{A'} \Theta^{B'}f_{A'} \tilde{\Theta}_{B'}\chi
 -  \tfrac{3}{2} f^{A'} \Theta^{B'}\tilde{f}_{A'} \tilde{\Theta}_{B'}\chi\nonumber\\
& + 3 \Theta^{B'}\tilde{f}^{A'} \tilde{\Theta}_{B'}\Theta_{A'}\chi .
\end{align}
\begin{align}\label{eq:Psi4CommEq2}
\Theta_{B'}\Theta_{A'}\tilde{\Theta}^{B'}\chi ={}&-6 \Lambda \Theta_{A'}\chi
 - 6 \Psi_{2} \Theta_{A'}\chi
 -  \tfrac{1}{2} \Theta_{A'}\mathcal{O}^{\dagger}{}\chi
 - 6 \chi \Theta_{A'}\Psi_{2}
 + 3 \tilde{f}_{B'} \Theta_{A'}\Theta^{B'}\chi\nonumber\\
& + 3 \Theta_{A'}\tilde{f}_{B'} \Theta^{B'}\chi
 -  \tfrac{3}{2} \tilde{f}_{B'} \sigma '^{B'} \tilde{\Theta}_{A'}\chi
 -  \tfrac{3}{2} \Theta_{B'}f^{B'} \tilde{\Theta}_{A'}\chi
 + \tfrac{3}{2} \tilde{\Theta}_{A'}\chi \tilde{\Theta}_{B'}\sigma '^{B'}\nonumber\\
& -  \sigma '^{B'} \tilde{\Theta}_{B'}\tilde{\Theta}_{A'}\chi .
\end{align}
\begin{align}\label{eq:Psi4CommEq3}
\Theta_{A'}\tilde{\Theta}^{A'}\Theta^{B'}\chi ={}&-4 \chi f^{B'} \Lambda
 - 2 \chi f^{B'} \Psi_{2}
 + 4 \chi \tilde{f}^{A'} \tilde{f}^{B'} \sigma '_{A'}
 - 2 f^{B'} \tilde{f}^{A'} \Theta_{A'}\chi
 - 2 f^{A'} \tilde{f}^{B'} \Theta_{A'}\chi\nonumber\\
& - 6 \chi \tilde{f}^{B'} \Theta_{A'}f^{A'}
 - 2 \chi \tilde{f}^{A'} \Theta_{A'}f^{B'}
 - 2 \chi f^{A'} \Theta_{A'}\tilde{f}^{B'}
 -  \tilde{f}^{B'} \Theta_{A'}\Theta^{A'}\chi
 -  f^{A'} \Theta_{A'}\tilde{\Theta}^{B'}\chi\nonumber\\
& + \Theta_{A'}\tilde{f}^{B'} \Theta^{A'}\chi
 - 2 f^{A'} \tilde{f}_{A'} \Theta^{B'}\chi
 - 2 \Lambda \Theta^{B'}\chi
 - 10 \Psi_{2} \Theta^{B'}\chi
 + 2 \chi \tilde{f}^{A'} \Theta^{B'}f_{A'}\nonumber\\
& - 2 \chi f^{A'} \Theta^{B'}\tilde{f}_{A'}
 + 3 \Theta^{A'}\chi \Theta^{B'}\tilde{f}_{A'}
 -  \tfrac{1}{2} \Theta^{B'}O^{\dagger}{}\chi
 - 10 \chi \Theta^{B'}\Psi_{2}
 - 3 \tilde{f}^{A'} \Theta^{B'}\Theta_{A'}\chi\nonumber\\
& - 2 \tilde{f}^{B'} \sigma '^{A'} \tilde{\Theta}_{A'}\chi
 + 4 \chi \tilde{f}^{B'} \tilde{\Theta}_{A'}\sigma '^{A'}
 -  \sigma '^{A'} \tilde{\Theta}_{A'}\tilde{\Theta}^{B'}\chi
 + \tfrac{3}{2} \tilde{f}^{A'} \sigma '_{A'} \tilde{\Theta}^{B'}\chi\nonumber\\
& -  \tfrac{5}{2} \Theta_{A'}f^{A'} \tilde{\Theta}^{B'}\chi
 + \tfrac{3}{2} \tilde{\Theta}_{A'}\sigma '^{A'} \tilde{\Theta}^{B'}\chi .
\end{align}
\begin{align}\label{eq:Psi4CommEq4}
\Theta_{B'}\Theta^{B'}\tilde{\Theta}_{A'}\chi ={}&\tfrac{1}{2} (-3 \tilde{f}_{B'} \sigma '^{B'} \tilde{\Theta}_{A'}\chi - 3 \Theta_{B'}f^{B'} \tilde{\Theta}_{A'}\chi + 3 \tilde{\Theta}_{A'}\chi \tilde{\Theta}_{B'}\sigma '^{B'} - 2 \sigma '^{B'} \tilde{\Theta}_{B'}\tilde{\Theta}_{A'}\chi).
\end{align}
\end{proposition}
\begin{proof}
The identities are verified by direct computation using the projected operators $\T_{A'},\tT_{A'}$ defined in \eqref{projectedTheta} and their commutator properties \eqref{eq:ProjCommutators}.
The Debye equation \eqref{eq:ThThtDebye} can be seen to coincide with \eqref{eq:Oadj} by GHP expansion.
\end{proof}

\begin{lemma} \label{lem:asdDebyeCurvature}
The ASD curvature components of a metric of the form \eqref{eq:MetricFromDeybe},\eqref{eq:MetricFromDeybeComps} are given by
\begin{subequations} 
\begin{align}
\dot\Psi_{0}={}& 0 \\
\dot\Psi_{1}={}& 0 \\
\dot\Psi_{2}={}& 0 \\
\dot\Psi_{3}={}&- \tfrac{1}{4} (\tau \tho - \rho \edt) \mathcal{O}^{\dagger}{}\chi, \label{eq:LinPsi3Final}\\
\dot\Psi_{4}={}&   -(\Psi_2^{4/3} \xi^a \Theta_a + 3 \Psi_{2}^2 + 6 \Psi_2 \Lambda) \chi - (\tfrac{1}{8}\mathcal{O}^{\dagger} - \rho\tho' + \tau\edt' + 2\Psi_2 + \tfrac{5}{2} \Lambda) \mathcal{O}^{\dagger}{}\chi .\label{eq:LinPsi4Final}
\end{align}
\end{subequations} 
\end{lemma}

\begin{proof}
The result is verified by direct computation using the projected operators $\T_{A'},\tT_{A'}$ defined in \eqref{projectedTheta} and their commutator properties \eqref{eq:ProjCommutators}. The first step is to derive an appropriate form for the components of the linearized ASD Weyl curvature spinor \eqref{linearizedWeyl}. Replacing \eqref{ooX} in \eqref{linearizedWeyl} and using \eqref{projectedTheta}, \eqref{projectedf}, \eqref{eq:FirstStructureEq}, we find
\begin{subequations}
\begin{align}
\dot \Psi_0 ={}&0, \label{linpsi0} \\
\dot \Psi_1 ={}&0, \label{linpsi1} \\
\dot \Psi_2 ={}&\frac{1}{12} (\tilde\Theta_{B'} + 2 \tilde f_{B'})(\tilde\Theta_{A'} + 2 \tilde f_{A'})X^{A'B'}, \label{linpsi2} \\
\dot \Psi_3 ={}&\frac{1}{8}\left((\Theta_{B'} +  f_{B'})(\tilde\Theta_{A'} + 2 \tilde f_{A'}) + (\tilde\Theta_{B'} + 3 \tilde f_{B'})\Theta_{A'} \right)X^{A'B'}, \label{linpsi3} \\
\dot \Psi_4 ={}&\frac{1}{2} \left( \Theta_{B'}\Theta_{A'} - \sigma'_{B'}(\tilde\Theta_{A'} + 2\tilde f_{A'}) \right)X^{A'B'}. \label{linpsi4}
\end{align}
\end{subequations}
So the first two equations, \eqref{linpsi0} and \eqref{linpsi1}, follow from the algebraic structure of \eqref{X}. For $\dot\Psi_2$, rewrite the operator in \eqref{linpsi2} using \eqref{tildef} and insert \eqref{X},
\begin{align}
 \dot\Psi_2 ={} \tfrac{1}{12} \phi^{-2}\tilde{\Theta}_{B'}\tilde{\Theta}_{A'}(\phi^{2}X^{A'B'}) 
 ={} \tfrac{1}{12} \phi^{-2}\tilde{\Theta}_{B'}\tilde{\Theta}_{A'}\tilde{\Theta}^{A'}(\phi^2\tilde{\Theta}^{B'}\chi)
 ={} 0,
\end{align}
where the last step follows from \eqref{squareTildeTheta}. Next we compute $\dot\psi_3$ by first inserting \eqref{X} into \eqref{linpsi3} and expanding out,
\begin{align}
\dot{\Psi}_{3}{}={}&- \tfrac{1}{8} \tilde{f}^{A'} \Theta_{A'}\tilde{\Theta}_{B'}\tilde{\Theta}^{B'}\chi
 -  \tfrac{1}{16} \Theta_{A'}\tilde{\Theta}_{B'}\tilde{\Theta}^{B'}\tilde{\Theta}^{A'}\chi
 + \tfrac{5}{8} \tilde{f}^{A'} \tilde{f}^{B'} \Theta_{B'}\tilde{\Theta}_{A'}\chi
 + \tfrac{3}{16} \tilde{f}^{A'} \Theta_{B'}\tilde{\Theta}_{A'}\tilde{\Theta}^{B'}\chi\nonumber\\
& -  \tfrac{1}{16} \Theta_{B'}\tilde{\Theta}_{A'}\tilde{\Theta}^{B'}\tilde{\Theta}^{A'}\chi
 + \tfrac{5}{16} \tilde{f}^{A'} \Theta_{B'}\tilde{\Theta}^{B'}\tilde{\Theta}_{A'}\chi
 + \tfrac{5}{8} \tilde{f}^{A'} \Theta_{B'}\tilde{f}^{B'} \tilde{\Theta}_{A'}\chi
 + \tfrac{1}{4} \Theta^{B'}\tilde{\Theta}^{A'}\chi \tilde{\Theta}_{A'}\tilde{f}_{B'}\nonumber\\
& -  \tfrac{1}{8} \tilde{f}^{A'} \tilde{\Theta}_{A'}\Theta_{B'}\tilde{\Theta}^{B'}\chi
 -  \tfrac{1}{16} \tilde{\Theta}_{A'}\Theta_{B'}\tilde{\Theta}^{B'}\tilde{\Theta}^{A'}\chi
 -  \tfrac{1}{8} f^{A'} \tilde{f}^{B'} \tilde{\Theta}_{A'}\tilde{\Theta}_{B'}\chi
 + \tfrac{1}{8} \Theta^{B'}\tilde{f}^{A'} \tilde{\Theta}_{A'}\tilde{\Theta}_{B'}\chi\nonumber\\
& -  \tfrac{1}{8} \tilde{\Theta}_{A'}\Theta_{B'}\tilde{f}^{B'} \tilde{\Theta}^{A'}\chi
 -  \tfrac{1}{4} f^{A'} \tilde{f}_{A'} \tilde{f}^{B'} \tilde{\Theta}_{B'}\chi
 + \tfrac{1}{4} \tilde{f}^{A'} \Theta_{A'}\tilde{f}^{B'} \tilde{\Theta}_{B'}\chi
 + \tfrac{1}{8} \Theta_{A'}\tilde{\Theta}^{B'}\tilde{f}^{A'} \tilde{\Theta}_{B'}\chi\nonumber\\
& -  \tfrac{3}{8} \tilde{f}^{A'} \Theta^{B'}\tilde{f}_{A'} \tilde{\Theta}_{B'}\chi
 + \tfrac{1}{8} \Theta^{B'}\tilde{\Theta}_{A'}\tilde{f}^{A'} \tilde{\Theta}_{B'}\chi
 -  \tfrac{1}{4} \Theta_{A'}\tilde{\Theta}^{A'}\chi \tilde{\Theta}_{B'}\tilde{f}^{B'}
 + \tfrac{1}{8} f^{A'} \tilde{\Theta}_{A'}\chi \tilde{\Theta}_{B'}\tilde{f}^{B'}\nonumber\\
& -  \tfrac{1}{8} \tilde{\Theta}^{A'}\chi \tilde{\Theta}_{B'}\Theta_{A'}\tilde{f}^{B'}
 -  \tfrac{1}{8} \tilde{f}^{A'} \tilde{\Theta}_{B'}\Theta_{A'}\tilde{\Theta}^{B'}\chi
 -  \tfrac{1}{16} \tilde{\Theta}_{B'}\Theta_{A'}\tilde{\Theta}^{B'}\tilde{\Theta}^{A'}\chi
 -  \tfrac{1}{8} \Theta^{B'}\tilde{f}^{A'} \tilde{\Theta}_{B'}\tilde{\Theta}_{A'}\chi\nonumber\\
& + \tfrac{1}{16} f^{A'} \tilde{\Theta}_{B'}\tilde{\Theta}_{A'}\tilde{\Theta}^{B'}\chi
 + \tfrac{1}{8} f^{A'} \tilde{f}_{A'} \tilde{\Theta}_{B'}\tilde{\Theta}^{B'}\chi
 -  \tfrac{1}{4} \Theta_{A'}\tilde{f}^{A'} \tilde{\Theta}_{B'}\tilde{\Theta}^{B'}\chi
 + \tfrac{1}{16} f^{A'} \tilde{\Theta}_{B'}\tilde{\Theta}^{B'}\tilde{\Theta}_{A'}\chi\nonumber\\
& + \tfrac{1}{8} f^{A'} \tilde{\Theta}_{B'}\tilde{f}_{A'} \tilde{\Theta}^{B'}\chi .
\end{align}
Using \eqref{eq:ThThtSimpSet} and \eqref{commutator-tildeTheta} yields
\begin{align}
\dot{\Psi}_{3}{}={}&\tfrac{7}{8} \tilde{f}^{A'} \tilde{f}^{B'} \Theta_{B'}\tilde{\Theta}_{A'}\chi
 + \tfrac{1}{2} \tilde{f}^{A'} \Theta_{B'}\tilde{\Theta}^{B'}\tilde{\Theta}_{A'}\chi
 + \tfrac{3}{2} \tilde{f}^{A'} \Lambda \tilde{\Theta}_{A'}\chi
 + \tfrac{9}{8} \tilde{f}^{A'} \Psi_{2} \tilde{\Theta}_{A'}\chi
 -  \tfrac{1}{8} \tilde{f}^{A'} \tilde{\Theta}_{A'}\Theta_{B'}\tilde{\Theta}^{B'}\chi\nonumber\\
& -  \tfrac{1}{16} \tilde{\Theta}_{A'}\Theta_{B'}\tilde{\Theta}^{B'}\tilde{\Theta}^{A'}\chi
 -  \tfrac{3}{8} f^{A'} \tilde{f}_{A'} \tilde{f}^{B'} \tilde{\Theta}_{B'}\chi
 + \tfrac{3}{8} \tilde{f}^{A'} \Theta_{A'}\tilde{f}^{B'} \tilde{\Theta}_{B'}\chi
 -  \tfrac{3}{8} \tilde{f}^{A'} \Theta^{B'}\tilde{f}_{A'} \tilde{\Theta}_{B'}\chi\nonumber\\
& -  \tfrac{1}{8} \tilde{\Theta}^{A'}\chi \tilde{\Theta}_{B'}\Theta_{A'}\tilde{f}^{B'}
 -  \tfrac{1}{8} \tilde{f}^{A'} \tilde{\Theta}_{B'}\Theta_{A'}\tilde{\Theta}^{B'}\chi
 -  \tfrac{1}{16} \tilde{\Theta}_{B'}\Theta_{A'}\tilde{\Theta}^{B'}\tilde{\Theta}^{A'}\chi
 -  \tfrac{1}{8} f^{A'} \tilde{f}^{B'} \tilde{\Theta}_{B'}\tilde{\Theta}_{A'}\chi .
\end{align}
To eliminate 4th order terms, we use \eqref{eq:Psi3CommEq1} and \eqref{eq:Psi3CommEq2}, leading to
\begin{align}
\dot{\Psi}_{3}{}={}&\tilde{f}^{A'} \tilde{f}^{B'} \Theta_{B'}\tilde{\Theta}_{A'}\chi
 + \tfrac{1}{2} \tilde{f}^{A'} \Theta_{B'}\tilde{\Theta}^{B'}\tilde{\Theta}_{A'}\chi
 + \tfrac{9}{8} \tilde{f}^{A'} \Lambda \tilde{\Theta}_{A'}\chi
 + \tfrac{33}{16} \tilde{f}^{A'} \Psi_{2} \tilde{\Theta}_{A'}\chi
 -  \tfrac{1}{8} \tilde{f}^{A'} \tilde{\Theta}_{A'}\Theta_{B'}\tilde{\Theta}^{B'}\chi\nonumber\\
& + \tfrac{3}{16} f^{A'} \tilde{f}^{B'} \tilde{\Theta}_{A'}\tilde{\Theta}_{B'}\chi
 + \tfrac{3}{8} \tilde{f}^{A'} \Theta_{A'}\tilde{f}^{B'} \tilde{\Theta}_{B'}\chi
 -  \tfrac{3}{8} \tilde{f}^{A'} \Theta^{B'}\tilde{f}_{A'} \tilde{\Theta}_{B'}\chi
 -  \tfrac{1}{8} \tilde{\Theta}^{A'}\chi \tilde{\Theta}_{B'}\Theta_{A'}\tilde{f}^{B'}\nonumber\\
& -  \tfrac{1}{8} \tilde{f}^{A'} \tilde{\Theta}_{B'}\Theta_{A'}\tilde{\Theta}^{B'}\chi
 -  \tfrac{1}{8} \tilde{f}^{A'} \tilde{\Theta}_{B'}\Theta^{B'}\tilde{\Theta}_{A'}\chi
 + \tfrac{1}{16} f^{A'} \tilde{f}^{B'} \tilde{\Theta}_{B'}\tilde{\Theta}_{A'}\chi
 -  \tfrac{3}{16} \tilde{f}^{A'} \tilde{\Theta}_{A'}f_{B'} \tilde{\Theta}^{B'}\chi\nonumber\\
& + \tfrac{3}{16} \tilde{f}^{A'} \tilde{\Theta}_{B'}f_{A'} \tilde{\Theta}^{B'}\chi .
\end{align}
Using the commutators \eqref{commutator-tildeTheta}, \eqref{commutator-tildeThetaTheta} and the Debye equation \eqref{eq:ThThtDebye} together with \eqref{eq:ThThtSimpSet} takes care of third order terms,
\begin{align}
\dot{\Psi}_{3}{}={}&- \tfrac{13}{8} \tilde{f}^{A'} \Lambda \tilde{\Theta}_{A'}\chi
 -  \tfrac{13}{16} \tilde{f}^{A'} \Psi_{2} \tilde{\Theta}_{A'}\chi
 -  \tfrac{1}{4} \tilde{f}^{A'} \tilde{\Theta}_{A'}\mathcal{O}^{\dagger}{}\chi
 + \tfrac{3}{8} \tilde{f}^{A'} \Theta_{A'}\tilde{f}^{B'} \tilde{\Theta}_{B'}\chi
 -  \tfrac{3}{8} \tilde{f}^{A'} \Theta^{B'}\tilde{f}_{A'} \tilde{\Theta}_{B'}\chi\nonumber\\
& + \tfrac{7}{16} \tilde{f}^{A'} \tilde{\Theta}_{A'}f_{B'} \tilde{\Theta}^{B'}\chi
 -  \tfrac{7}{16} \tilde{f}^{A'} \tilde{\Theta}_{B'}f_{A'} \tilde{\Theta}^{B'}\chi .
\end{align}
\begin{subequations}
Finally, the irreducible decompositions
\begin{align}
\tilde{\Theta}_{A'}f_{B'}={}&\tilde{\Theta}_{(A'}f_{B')}
 + \tfrac{1}{2} \bar\epsilon_{A'B'} \tilde{\Theta}_{C'}f^{C'}, \\
\Theta_{A'}\tilde{f}_{B'}={}&\Theta_{(A'}\tilde{f}_{B')}
 + \tfrac{1}{2} \bar\epsilon_{A'B'} \Theta_{C'}\tilde{f}^{C'}, \label{eq:IrrDecThft}
\end{align}
\end{subequations}
lead to
\begin{align}
\dot{\Psi}_{3}{}={}&- \tfrac{1}{4} \tilde{f}^{A'} \tilde{\Theta}_{A'}\mathcal{O}^{\dagger}{}\chi,
\end{align}
which gives \eqref{eq:LinPsi3Final} by GHP expanding $ \tilde{f}^{A'} \tilde{\Theta}_{A'}$.

To compute $\dot\Psi_4$, insert \eqref{X} into \eqref{linpsi4} and use the identity \eqref{eq:Psi4CommEq1} together with the Debye equation \eqref{eq:ThThtDebye} and \eqref{eq:ThThtSimpSet} leading to
\begin{align}
\dot{\Psi}_{4}{}={}&\tfrac{9}{2} \chi f^{A'} \tilde{f}_{A'} \Lambda
 - 18 \chi \Lambda^2
 + \tfrac{3}{8} f^{A'} \tilde{f}_{A'} \mathcal{O}^{\dagger}{}\chi
 - 2 \Lambda \mathcal{O}^{\dagger}{}\chi
 -  \tfrac{1}{8} \mathcal{O}^{\dagger}{}\mathcal{O}^{\dagger}{}\chi
 + \tfrac{9}{2} \chi f^{A'} \tilde{f}_{A'} \Psi_{2}
 + 6 \chi \Lambda \Psi_{2}\nonumber\\
& -  \tfrac{7}{4} \mathcal{O}^{\dagger}{}\chi \Psi_{2}
 + \tfrac{15}{2} \chi \Psi_{2}^2
 + \tfrac{21}{2} \tilde{f}^{A'} \Psi_{2} \Theta_{A'}\chi
 + 6 \chi \Lambda \Theta_{A'}\tilde{f}^{A'}
 -  \tfrac{3}{2} \chi \Psi_{2} \Theta_{A'}\tilde{f}^{A'}
 -  \tfrac{1}{2} \tilde{f}^{A'} \Theta_{A'}\mathcal{O}^{\dagger}{}\chi\nonumber\\
& + \tfrac{21}{2} \chi \tilde{f}^{A'} \Theta_{A'}\Psi_{2}
 + \tfrac{3}{2} \chi \tilde{f}^{A'} \Theta_{A'}\Theta_{B'}\tilde{f}^{B'}
 + \tfrac{3}{4} f^{A'} \tilde{f}^{B'} \Theta_{A'}\tilde{\Theta}_{B'}\chi
 + \tfrac{9}{4} f^{A'} \tilde{f}_{A'} \tilde{f}^{B'} \Theta_{B'}\chi\nonumber\\
& + \tfrac{3}{2} \tilde{f}^{A'} \Theta_{A'}\chi \Theta_{B'}\tilde{f}^{B'}
 + 3 \tilde{f}^{A'} \tilde{f}^{B'} \Theta_{B'}\Theta_{A'}\chi
 -  \tfrac{1}{2} \tilde{f}^{A'} \Theta_{B'}\Theta_{A'}\tilde{\Theta}^{B'}\chi
 -  \tfrac{1}{2} \tilde{f}^{A'} \Theta_{B'}\Theta^{B'}\tilde{\Theta}_{A'}\chi\nonumber\\
& + \tfrac{3}{4} f^{A'} \tilde{f}^{B'} \Theta_{B'}\tilde{\Theta}_{A'}\chi
 + \tfrac{3}{2} \tilde{f}^{A'} \Theta_{B'}\tilde{\Theta}^{B'}\Theta_{A'}\chi
 - 3 \tilde{f}^{A'} \Theta_{A'}\tilde{f}_{B'} \Theta^{B'}\chi
 + \Theta_{A'}\tilde{\Theta}_{B'}\chi \Theta^{B'}\tilde{f}^{A'}\nonumber\\
& + \tfrac{1}{2} \Theta_{B'}\tilde{\Theta}_{A'}\chi \Theta^{B'}\tilde{f}^{A'}
 + \tfrac{3}{2} f^{A'} \Lambda \tilde{\Theta}_{A'}\chi
 + \tfrac{3}{4} f^{A'} \Psi_{2} \tilde{\Theta}_{A'}\chi
 -  \Theta^{A'}\Psi_{2} \tilde{\Theta}_{A'}\chi
 + \tfrac{3}{4} \tilde{f}^{A'} \Theta_{B'}f^{B'} \tilde{\Theta}_{A'}\chi\nonumber\\
& -  \tfrac{3}{2} \tilde{f}^{A'} \tilde{f}^{B'} \sigma '_{A'} \tilde{\Theta}_{B'}\chi
 + \tfrac{3}{4} \tilde{f}^{A'} \Theta_{A'}f^{B'} \tilde{\Theta}_{B'}\chi
 + \tfrac{3}{4} f^{A'} \Theta_{A'}\tilde{f}^{B'} \tilde{\Theta}_{B'}\chi
 + \tfrac{1}{2} \Theta_{A'}\Theta^{B'}\tilde{f}^{A'} \tilde{\Theta}_{B'}\chi\nonumber\\
& -  \tfrac{3}{4} \tilde{f}^{A'} \Theta^{B'}f_{A'} \tilde{\Theta}_{B'}\chi
 + \tfrac{3}{4} f^{A'} \Theta^{B'}\tilde{f}_{A'} \tilde{\Theta}_{B'}\chi
 -  \tfrac{3}{2} \Theta^{B'}\tilde{f}^{A'} \tilde{\Theta}_{B'}\Theta_{A'}\chi
 + \tfrac{1}{2} \tilde{f}^{A'} \sigma '^{B'} \tilde{\Theta}_{B'}\tilde{\Theta}_{A'}\chi\nonumber\\
& -  \tfrac{1}{4} \sigma '^{A'} \tilde{\Theta}_{B'}\tilde{\Theta}_{A'}\tilde{\Theta}^{B'}\chi .
\end{align}
To convert third order terms we use identities \eqref{eq:Psi4CommEq2}, \eqref{eq:Psi4CommEq3}, \eqref{eq:Psi4CommEq4}. After a commutator  \eqref{commutator-tildeThetaTheta} and \eqref{eq:ThThtSimpSet} is used, we have
\begin{align}
\dot{\Psi}_{4}{}={}&\tfrac{9}{2} \chi f^{A'} \tilde{f}_{A'} \Lambda
 + 6 \chi \Lambda^2
 + \tfrac{3}{8} f^{A'} \tilde{f}_{A'} \mathcal{O}^{\dagger}{}\chi
 - 2 \Lambda \mathcal{O}^{\dagger}{}\chi
 -  \tfrac{1}{8} \mathcal{O}^{\dagger}{}\mathcal{O}^{\dagger}{}\chi
 + \tfrac{9}{2} \chi f^{A'} \tilde{f}_{A'} \Psi_{2}
 + 3 \chi \Lambda \Psi_{2}\nonumber\\
& -  \tfrac{7}{4} \mathcal{O}^{\dagger}{}\chi \Psi_{2}
 + 3 \tilde{f}^{A'} \Lambda \Theta_{A'}\chi
 -  \tilde{f}^{A'} \Theta_{A'}\mathcal{O}^{\dagger}{}\chi
 -  \tfrac{3}{4} f^{A'} \tilde{f}^{B'} \Theta_{A'}\tilde{\Theta}_{B'}\chi
 + \tfrac{9}{4} f^{A'} \tilde{f}_{A'} \tilde{f}^{B'} \Theta_{B'}\chi\nonumber\\
& + \tfrac{3}{4} f^{A'} \tilde{f}^{B'} \Theta_{B'}\tilde{\Theta}_{A'}\chi
 -  \tfrac{1}{2} \Theta_{A'}\tilde{\Theta}_{B'}\chi \Theta^{B'}\tilde{f}^{A'}
 + \tfrac{1}{2} \Theta_{B'}\tilde{\Theta}_{A'}\chi \Theta^{B'}\tilde{f}^{A'}
 + \tfrac{3}{2} f^{A'} \Lambda \tilde{\Theta}_{A'}\chi\nonumber\\
& + \tfrac{3}{4} f^{A'} \Psi_{2} \tilde{\Theta}_{A'}\chi
 -  \Theta^{A'}\Psi_{2} \tilde{\Theta}_{A'}\chi
 -  \tfrac{3}{2} \tilde{f}^{A'} \Theta_{B'}f^{B'} \tilde{\Theta}_{A'}\chi
 -  \tfrac{3}{4} \tilde{f}^{A'} \tilde{f}^{B'} \sigma '_{A'} \tilde{\Theta}_{B'}\chi\nonumber\\
& + \tfrac{3}{4} \tilde{f}^{A'} \Theta_{A'}f^{B'} \tilde{\Theta}_{B'}\chi
 -  \tfrac{3}{4} f^{A'} \Theta_{A'}\tilde{f}^{B'} \tilde{\Theta}_{B'}\chi
 + \tfrac{1}{2} \Theta_{A'}\Theta^{B'}\tilde{f}^{A'} \tilde{\Theta}_{B'}\chi
 -  \tfrac{3}{4} \tilde{f}^{A'} \Theta^{B'}f_{A'} \tilde{\Theta}_{B'}\chi\nonumber\\
& + \tfrac{3}{4} f^{A'} \Theta^{B'}\tilde{f}_{A'} \tilde{\Theta}_{B'}\chi
 + \tfrac{3}{4} \tilde{f}^{A'} \tilde{\Theta}_{A'}\chi \tilde{\Theta}_{B'}\sigma '^{B'}
 -  \tfrac{1}{4} \sigma '^{A'} \tilde{\Theta}_{B'}\tilde{\Theta}_{A'}\tilde{\Theta}^{B'}\chi .
\end{align}
Now, the commutator  \eqref{commutator-tildeTheta} together with the irreducible decomposition
\begin{align}
\Theta_{A'}\tilde{\Theta}_{B'}\chi ={}&\Theta_{(A'}\tilde{\Theta}_{B')}\chi
 + \tfrac{1}{2} \bar\epsilon_{A'B'} \Theta_{C'}\tilde{\Theta}^{C'}\chi,
\end{align}
and \eqref{eq:ThThtDebye}, \eqref{eq:ThThtSimpSet} yields
\begin{align}
\dot{\Psi}_{4}{}={}&- \tfrac{5}{2} \Lambda \mathcal{O}^{\dagger}{}\chi
 -  \tfrac{1}{8} \mathcal{O}^{\dagger}{}\mathcal{O}^{\dagger}{}\chi
 - 6 \chi \Lambda \Psi_{2}
 - 2 \mathcal{O}^{\dagger}{}\chi \Psi_{2}
 - 3 \chi \Psi_{2}^2
 -  \tfrac{3}{2} \tilde{f}^{A'} \Psi_{2} \Theta_{A'}\chi
 -  \tilde{f}^{A'} \Theta_{A'}\mathcal{O}^{\dagger}{}\chi\nonumber\\
& + \tfrac{3}{2} f^{A'} \Lambda \tilde{\Theta}_{A'}\chi
 + \tfrac{3}{4} f^{A'} \Psi_{2} \tilde{\Theta}_{A'}\chi
 -  \tfrac{1}{2} \Theta^{A'}\Psi_{2} \tilde{\Theta}_{A'}\chi
 -  \tfrac{5}{4} \tilde{f}^{A'} \Theta_{B'}f^{B'} \tilde{\Theta}_{A'}\chi\nonumber\\
& -  \tfrac{1}{2} \tilde{f}^{A'} \tilde{f}^{B'} \sigma '_{A'} \tilde{\Theta}_{B'}\chi
 + \tfrac{3}{4} \tilde{f}^{A'} \Theta_{A'}f^{B'} \tilde{\Theta}_{B'}\chi
 -  \tfrac{3}{4} f^{A'} \Theta_{A'}\tilde{f}^{B'} \tilde{\Theta}_{B'}\chi
 -  \tfrac{3}{4} \tilde{f}^{A'} \Theta^{B'}f_{A'} \tilde{\Theta}_{B'}\chi\nonumber\\
& + \tfrac{3}{4} f^{A'} \Theta^{B'}\tilde{f}_{A'} \tilde{\Theta}_{B'}\chi
 + \tfrac{1}{2} \tilde{f}^{A'} \tilde{\Theta}_{A'}\chi \tilde{\Theta}_{B'}\sigma '^{B'}.
\end{align}
The Ricci identity
\begin{align}
2 \Psi_{3} -  \tilde{f}_{B'} \sigma '^{B'} + \Theta_{A'}f^{A'} -  \tilde{\Theta}_{B'}\sigma '^{B'}={}&0.
\end{align}
and the irreducible decompositions \eqref{eq:IrrDecThft} and
\begin{align}
\Theta_{A'}f_{B'}={}&\Theta_{(A'}f_{B')}
 + \tfrac{1}{2} \bar\epsilon_{A'B'} \Theta_{C'}f^{C'},
\end{align}
lead to
\begin{align}
\dot{\Psi}_{4}{}={}&- \tfrac{5}{2} \Lambda \mathcal{O}^{\dagger}{}\chi
 -  \tfrac{1}{8} \mathcal{O}^{\dagger}{}\mathcal{O}^{\dagger}{}\chi
 - 6 \chi \Lambda \Psi_{2}
 - 2 \mathcal{O}^{\dagger}{}\chi \Psi_{2}
 - 3 \chi \Psi_{2}^2
 -  \tfrac{3}{2} \tilde{f}^{A'} \Psi_{2} \Theta_{A'}\chi
 -  \tilde{f}^{A'} \Theta_{A'}\mathcal{O}^{\dagger}{}\chi\nonumber\\
& + \tilde{f}^{A'} \Psi_{3} \tilde{\Theta}_{A'}\chi
 -  \tfrac{1}{2} \Theta^{A'}\Psi_{2} \tilde{\Theta}_{A'}\chi .
\end{align}
To bring it into the final form we use the Bianchi identity
\begin{align}
-3 f_{A'} \Psi_{2} + 2 \tilde{f}_{A'} \Psi_{3} + \Theta_{A'}\Psi_{2} -  \tilde{\Theta}_{A'}\Psi_{3}={}&0,
\end{align}
resulting in
\begin{align}
\dot{\Psi}_{4}{}={}&
- \tfrac{3}{2}\Psi_2 ( f^{A'} \tilde{\Theta}_{A'} + \tilde{f}^{A'} \Theta_{A'} + 2 \Psi_{2} + 4 \Lambda)\chi
 + \tfrac{1}{2} (\tilde{\Theta}_{A'} - 4\tilde{f}_{A'})\Psi_{3} \tilde{\Theta}^{A'}\chi \nonumber \\
 {}& -\left(\tfrac{1}{8} \mathcal{O}^{\dagger}{} +\tilde{f}^{A'} \Theta_{A'} + 2 \Psi_{2} + \tfrac{5}{2} \Lambda  \right)\mathcal{O}^{\dagger}{}\chi.
\end{align}
GHP expansion of $ \tilde{f}_{A'},  \Theta_{A'},  \tilde{\Theta}_{A'}$ leads to 
\begin{align}
\dot \Psi_{4}={}&  -  \tfrac{3}{2} \Psi_{2} (\rho' \tho - \rho \tho' + \tau \edt' - \tau' \edt + 2 \Psi_2 + 4 \Lambda)\chi  + \tfrac{1}{2} ((\tho - 4\rho) \Psi_{3}) \edt \chi - \tfrac{1}{2} ((\edt -4\tau) \Psi_{3}) \tho \chi \nonumber \\
& - (\tfrac{1}{8}\mathcal{O}^{\dagger} - \rho\tho' + \tau\edt' + 2\Psi_2 + \tfrac{5}{2} \Lambda) \mathcal{O}^{\dagger}{}\chi .
\end{align}
Comparison to the projected Killing vector defined in \eqref{eq:xiGHP} shows
\eqref{eq:LinPsi4Final}.
\end{proof}

\begin{proof}[Proof of Theorem~\ref{thm}]
The result follows from Lemma~\ref{lem:asdDebyeCurvature} by imposing the Debye equation \eqref{eq:DebyeEquation1}.
\end{proof}

\begin{remark}
In the special case of vacuum Petrov type D and for tetrads invariant under $\xi^a$, \eqref{eq:LinPsi4Final} reduces to $\dot \Psi_{4} = \xi^a \nabla_a \chi$, 
see \cite{aksteiner:thesis}.
\end{remark}

\begin{remark} \label{rem:RelatedWork}
Let us finally compare to three references closely related to the results of this section.
\begin{enumerate}
\item In \cite{CK1979} Kegeles and Cohen discuss Debye potentials for algebraically special geometries. They restricted to vacuum Petrov type D for the derivation of the linearized Weyl spinor, see equation (5.28) in that reference. They reduced the ASD Weyl curvature to type N and also the $\dot\Psi_4$ to first order.
 \item In \cite{1979Stewart}, Stewart derived the linearized connection and curvature components for vacuum type II perturbations in terms of a Debye potential. However, the result was not fully simplified, see equation (4.27) of that reference, so that the type N property could not be observed. The result was also presented in terms of a real metric, which means that all terms involving $\bar\chi$ correspond to self dual Weyl curvature, while $\chi$ terms belong to anti-self dual Weyl curvature.
 
 It should also be noted that in general, linearized Dyad components differ from dyad components of the linearized field. In this paper $\dot\Psi_i$ refers to the latter, while Stewart used the linearized Newman-Penrose equations, i.e. the former. However, he made a special choice of linearized tetrad for which the two sets of linearized Weyl components coincide. 
\item In \cite{Jeffryes:1986}, Jeffryes discusses a reduction to scalar potentials for algebraically special solutions to the full non-linear Einstein-Yang-Mills equations. Further it is shown that, to linear order, this construction reduces to the Debye potential formulation. Remarkably, the ASD Weyl curvature can be simplified already on the non-linear level, so that Theorem~\ref{thm} we discuss here follows from the linearized equations (8.75-77) of that reference.
\end{enumerate}
\end{remark} 

\subsection*{Acknowledgements}
This work started while the authors were in residence at Institut Mittag-Leffler in Djursholm, Sweden during the fall of 2019, supported by the Swedish Research Council under grant no. 2016-06596. SA thanks Thomas B\"ackdahl for xAct support and Benjamin Jeffryes for comments on his preprint \cite{Jeffryes:1986}. BA is supported by a postdoctoral fellowship from Conicet (Argentina). BFW acknowledges support from NSF grant PHY 1607323, sabbatical support from the University of Florida and the Observaroire de Paris at Meudon, and the Institut d'Astrophysique de Paris.

\appendix 

\section{Robinson-Trautman metrics} \label{app:RobinsonTrautman}

In 1968, Robinson and Trautman, \cite{1962RobinsonTrautman}, published a line element for which the vacuum Einstein equations 
reduce essentially to a non-linear fourth order equation for a real scalar function. It admits solutions of all Petrov types. 
Here we briefly review the reduction and also the explicit example of Petrov type II given in \cite{1962RobinsonTrautman}.

In coordinates $(u, r, \zeta, \bar{\zeta})$ and with real functions $H,P$, define the tetrad
\begin{align} \label{eq:RTtetrad1}
l^{a}={}(\partial_r)^{a}, \qquad
n^{a}={}(\partial_u)^{a} - H (\partial_r)^{a},\qquad
m^{a}={}P r^{-1}(\partial_\zeta)^{a}, \qquad
\bar{m}^{a}={} P r^{-1}(\partial_{\bar{\zeta}})^{a}.
\end{align}
Due to the normalization $l^a n_a = 1, m^a \bar{m}_a = -1$, the inverse is given by
\begin{align}
l_{a}={}\dd u_{a}, \qquad
n_{a}={}\dd r_{a} + H \dd u_{a} , \qquad
m_{a}={}- r P^{-1} \dd \zeta_{a}, \qquad
\bar{m}_{a}={}- r P^{-1}\dd\bar{\zeta}_{a},
\end{align}
so that the metric $g_{ab} = 2l_{(a} n_{b)} - 2 m_{(a} \bar{m}_{b)}$ is of the form
\begin{align}
g_{ab} ={}&
 2\dd r_{(a} \dd u_{b)}
 +2 H \dd u_{a} \dd u_{b}
 -2r^2 P^{-2}\dd \zeta_{(a} \dd \bar{\zeta}_{b)}.
\end{align}
This metric is Ricci flat if $H$ is given by
\begin{align} \label{eq:RTHdef}
H={}& P^2 \partial_\zeta \partial_{\bar{\zeta}}\log(P)
 - r \partial_u \log(P) - m r^{-1}, \qquad \text{with } m = m(u),
\end{align}
and $P$ being independent of $r$, solving
\begin{align}
P^2 \partial_\zeta \partial_{\bar{\zeta}} \left(P^2 \partial_\zeta  \partial_{\bar{\zeta}}\log(P)\right) -  \partial_u m + 3 m \partial_u \log(P) ={}&0.
\end{align}
For the connection and curvature we find
\begin{align}
\kappa = \sigma = \sigma' = \tau = \tau' = \epsilon = 0, \qquad R_{ab}=0, \qquad \Psi_0 = \Psi_1 = 0,
\end{align}
in particular the metric is algebraically special. The non-vanishing spin coefficients are given by
\begin{subequations} 
\begin{align}
\kappa'={}&- P r^{-1} \partial_{\bar{\zeta}} H,&&&
\rho ={}&- r^{-1},&&&
\rho'={}& H r^{-1} \partial_u \log(P),\\
\gamma ={}&\tfrac{1}{2} \partial_r H,&&&
\alpha ={}& \tfrac{1}{2} r^{-1} \partial_{\bar{\zeta}} P,&&&
\beta ={}&- \tfrac{1}{2} r^{-1} \partial_{\zeta} P,
\end{align}
\end{subequations} 
with $H$ given in \eqref{eq:RTHdef} and the remaining Weyl components are of the form
\begin{subequations} 
\begin{align}
\Psi_2 ={}& -mr^{-3}, \qquad
\Psi_3 = - P r^{-2} \partial_{\bar{\zeta}} \left( P^2 \partial_\zeta \partial_{\bar{\zeta}} \log(P) \right), \\
\Psi_4 ={}& r^{-2} \partial_{\bar{\zeta}} \left( P^2 \partial_{\bar{\zeta}} \left( P^2 \partial_{\zeta}\partial_{\bar{\zeta}}\log(P) -r\partial_u \log(P)  \right) \right).
\end{align}
\end{subequations}

An explicit example of a type II geometry, found in \cite{1962RobinsonTrautman}, is given by
\begin{align}
 P = (\zeta + \bar{\zeta})^{3/2}, \qquad m = \text{const.},
\end{align}
leading to $H = - m/r - 3(\zeta + \bar{\zeta})/2$ and to the metric
\begin{align} \label{eq:TypeIIRTsol}
g_{ab} ={}&
 2\dd r_{(a} \dd u_{b)}
 -\left(3(\zeta + \bar{\zeta}) + 2m/r \right) \dd u_{a} \dd u_{b}
 -2r^2 (\zeta + \bar{\zeta})^{-3} \dd \zeta_{(a} \dd \bar{\zeta}_{b)},
\end{align}
see also \cite[\S 28]{stephani_kramer_maccallum_hoenselaers_herlt_2003}. The spin coefficients reduce to
\begin{align}
\kappa' =  \frac{3 (\zeta +\bar{\zeta})^{3/2}}{2 r}, &&
\rho = - \frac{1}{r},&&
\rho' = -\frac{m}{r^2} -\frac{3 (\zeta +\bar{\zeta})}{2 r},&&
\gamma = \frac{m}{2 r^2},&&
\alpha = - \beta = \frac{3 (\zeta +\bar{\zeta})^{1/2}}{4 r},
\end{align}
and the curvature components are given by
\begin{align}
\Psi_{2} = - \frac{m}{r^3},&&
\Psi_{3} = \frac{3 (\zeta +\bar{\zeta})^{3/2}}{2 r^2},&&
\Psi_{4} = - \frac{9 (\zeta +\bar{\zeta})^2}{2 r^2}.
\end{align}

\newcommand{\prd}{Phys. Rev. D} 
\newcommand{\apj}{Astrophysical J.}

\bibliographystyle{plain}

\end{document}